\documentclass[10pt, twocolumn]{IEEEtran}
\usepackage[dvips]{graphicx}
\usepackage{subfigure}
\usepackage{cite}
\usepackage{epsfig, psfrag, amsmath, amssymb, amsfonts, latexsym, eucal, balance, hyperref,indentfirst,bookmark}
\newtheorem{thm}{Theorem}%[section]
\newtheorem{cor}{Corollary}

\newtheorem{prop}{Proposition}

\begin{document}
%%%%%%%%%%%%%%%%%%%%%%%%%%%%%%%%%%%%%%%%%%%%%%%%%%%%%%%%%%%%%%%%%%%%%%%%%%%%%%%%%%%%%%%%%%%%%%%%%%%%%%%%%%%%
\title{Managing Interference Correlation Through Random Medium Access}
\author{{Yi Zhong, Wenyi Zhang, \emph{Senior Member, IEEE} and
 Martin Haenggi, \emph{Senior Member, IEEE}}
 \thanks{Y. Zhong and W. Zhang are with Department of Electronic Engineering and Information Science, University of Science and Technology of China, Hefei 230027, China (email: geners@mail.ustc.edu.cn, wenyizha@ustc.edu.cn).
 M. Haenggi is with Department of Electrical Engineering, University of Notre Dame, Notre Dame, IN 46556, USA (email: mhaenggi@nd.edu).
%The research has been supported by MIIT of China through grant 2011ZX03001-006-01 and by the Doctorial Program of Higher Education of China through grant 20103402120023.
The research has been supported by the National Basic
Research Program of China (973 Program) through grant
2012CB316004, National Natural Science Foundation of
China through grant 61071095, MIIT of China through
grant 2011ZX03001-006-01, and by the US NSF through grant CCF 1216407.
}}
\maketitle
\thispagestyle{empty}
\begin{abstract}
The capacity of wireless networks is fundamentally limited by interference. However, little research has focused on the interference correlation, which may greatly increase the local delay (namely the number of time slots required for a node to successfully transmit a packet). This paper focuses on the question whether increasing randomness in the MAC, specifically frequency-hopping multiple access (FHMA) and ALOHA, helps to reduce the effect of interference correlation. We derive closed-form results for the mean and variance of the local delay for the two MAC protocols and evaluate the optimal parameters that minimize the mean local delay. Based on the optimal parameters, we identify two operating regimes, the correlation-limited regime and the bandwidth-limited regime. Our results reveal that while the mean local delays for FHMA with $N$ sub-bands and for ALOHA with transmit probability $p$ essentially coincide when $p=\frac{1}{N}$, a fundamental discrepancy exists between their variances.
We also discuss implications from the analysis, including an interesting mean delay-jitter tradeoff, and convenient bounds on the tail probability of the local delay, which shed useful insights into system design.
\end{abstract}
%\newpage
\setcounter{page}{1}
%The results reveal that even simply employing these kinds of randomness will greatly decrease the local delay.

\section{Introduction}
\subsection{Motivation}
%$$-\frac{\pi\sqrt{S}}{4}-\sqrt{S}\arctan({2\sqrt{S}})+\frac{S\sqrt{S}}{9S^2+10S+1}\left(\pi+3\pi S+2\sqrt{S}\ln\frac{S}{4S+1}-(3S-1)\arctan({2\sqrt{S}})\right)$$
%$$-\frac{\pi\sqrt{S}}{4}-\frac{\sqrt{S}}{3}\arctan\frac{2{\sqrt{S}}}{3}+\frac{3S\sqrt{S}}{25S^2+234S+81}\left(3\pi-\frac{5}{3}\pi S+6\sqrt{S}\ln\frac{9S}{4S+9}+(5S+9)\arctan\frac{2{\sqrt{S}}}{3}\right)$$
A main limitation to the capacity of wireless communication systems is interference, which depends upon a number of factors, including the locations of interfering transmitters.
The issue of interference has been studied extensively in the literature; however, much less attention has been paid to the topic of interference correlation until recently.
Interference correlation generally captures the fact that the interference created by interfering transmitters is a correlated stochastic process both spatially and temporally.
It is well recognized that correlated fading reduces the performance gain in multi-antenna communications \cite{chuah2002capacity}.
Likewise, it has recently been also proved that interference correlation decreases the diversity gain \cite{net:Haenggi12cl}\cite{net:Haenggi14twc}.
Interference correlation partially comes come from correlated channel attenuation, like correlated fading and shadowing,
but more importantly, such correlation stems from the spatial distribution of transmitters and the MAC protocols since they determine the locations and the active pattern of the interferers, which then determine the structure of the interference.
The lines of recent research can be divided into three categories based on different configurations for the receiver:
\begin{itemize}
\item \textbf{Correlation between different time slots:}
Assume that the receiver is equipped with a single antenna.
This line of research explored the interference correlation at the same receiver between different time slots.
Related works include \cite{schilcher2011temporal, ganti2009spatial, net:Haenggi12tit, net:Haenggi14twc}. %\cite{schilcher2011temporal} and \cite{ganti2009spatial}.
\item \textbf{Correlation between different receive antennas:}
Assume that the receiver is equipped with co-located multiple antennas.
The correlation between different antennas exists because the interferences received by different antennas come from the same source of transmitters. Related works include \cite{net:Haenggi12cl}.
\item \textbf{Correlation between different receivers:}
This refers to the interference correlation between different receivers which are separated (a few wavelengths apart).
Since the network may make use of relay and cooperative transmission, it is necessary to consider this type of interference correlation for an accurate analysis. Related works include \cite{ganti2009spatial}.
\end{itemize}

In this work, we focus on the interference correlation between different time slots at the same receiver, i.e., the temporal correlation.
The interference power constitutes a stochastic process, wherein the randomness comes from three sources: the spatial distribution of nodes, the fading and the MAC.
The interferences at two different time slots are correlated because they come from correlated sets of transmitters and the fading, shadowing and traffic may also be correlated.
In this paper, we only focus on the correlation caused by the spatial distribution of transmitters and the MAC, assuming that fading and shadowing are independent.
This type of correlation brings about the fact that if transmission fails in a previous time slot, there is a significant probability that the subsequent transmission will also fail in the next few time slots \cite{net:Haenggi14twc}\cite{ganti2009spatial}.
Thus a simple retransmission mechanism may not be an effective method.
The most direct impact of this type of correlation is the increase of the local delay.
Local delay is defined as the number of time slots required by a node to successfully transmit a packet to its next-hop node\footnote{The definition of local delay in our work is consistent with \cite{baccelli2010new}. In some other works, like \cite{net:Haenggi12tit}, the local delay denotes the \emph{mean} number of time slots required to successfully transmit a packet.}.
%With temporal correlation between different time slots, the mean local delay will inevitably be larger than in the completely independent case.

As a motivating example, consider a spatial network without mobility or fading and without a MAC coordinating.
Hence the interference power experienced by a receiver remains fixed for all time slots; it is a randomly variable uniquely determined by the spatial distribution of nodes.
The local delay, as a random variable, in that extreme case is two-valued: either one frame (good realization of the spatial distribution of nodes) or infinite (bad realization of the spatial distribution of nodes).
In this case, the transmission success events are fully correlated (one success implies success in each time slot, and vice versa), and the mean local delay is infinite.

In view of this, we consider some forms of man-made randomization by introducing MAC dynamics to reduce the interference correlation.
The following analysis will be carried out in parallel under two different kinds of MAC protocols:
\begin{itemize}
\item \textbf{FHMA (Frequency-hopping multiple access):}
FHMA is implemented by simply dividing the entire frequency band into $N$ sub-bands and letting each transmitter independently choose a sub-band uniformly randomly in each time slot.
We focus on slow frequency-hopping, i.e., hopping at the time scale of a time slot, not at the time scale of a symbol.
%The frequency-hopping in our work indicates the \emph{slow} frequency-hopping, which refers to hopping on the order of a time slot, not on the order of a symbol time.
There are three benefits by splitting the entire frequency bands into sub-bands.
First and foremost, it increases the uncertainty in the active pattern of interfering nodes, thereby reducing the effect of interference correlation.
Second, the interference for a given transmission is also reduced because the intensity of the interfering transmitters are scaled by $\frac{1}{N}$.
Third, the noise power is also scaled by $\frac{1}{N}$ since each transmission occurs in a narrow sub-band.
Meanwhile, on the other side, splitting into sub-bands scales down the rate.
\item \textbf{ALOHA:}
In ALOHA, if a packet is to be transmitted during a time slot, the packet will only be transmitted with a certain probability using the entire frequency band. Decreasing the transmit probability increases the uncertainty in the active pattern of interfering nodes and reduces the interference, while the noise power will not be reduced.
Meanwhile, transmitting probabilistically scales down the rate.
%a packet will less likely be transmitted in each time slot, which may increase the local delay in order to successfully transmit that packet.
\end{itemize}

Since FHMA is often viewed as a spread-spectrum technique, we briefly comment on DS-CDMA.
For synchronous orthogonal CDMA like those using Walsh codes, a receiver can in theory completely reject arbitrarily strong signals from interfering transmitters using different spreading sequences; thus, only those transmitters using the same spreading sequence as the desired link will cause interference.
If the spreading sequence is randomly chosen for each transmission, the analysis and results of the local delay are exactly the same as that for FHMA.
%Because in synchronous CDMA, both the desired signal and the interference are enhanced by the same spread spectrum process gain at the receiver, leading to the same expressions for SINR as in FHMA.
For asynchronous CDMA using pseudo-noise (PN) sequences, the interference comes from all transmitters and is usually approximated as Gaussian noise in the literature.
The works in \cite{net:Andrews07commag} and \cite{weber2005transmission} have discussed the difference between asynchronous CDMA and FHMA in terms of outage probability and throughput.
In asynchronous CDMA, although the desired signal is increased by the processing gain, the interference still comes from all transmitters.
%Comparing with FHMA, the main difference of asynchronous CDMA lies in that there is a scaling change of the SINR due to the processing gain.
Therefore, the analysis of the local delay is similar as that for FHMA with $N=1$, i.e., no bandwidth splitting is employed.
We will show that in this case the distribution of the local delay has a heavy tail, which results in an infinite mean local delay.

\subsection{Related Works}
%%In \cite{chuah2002capacity}, the authors assume that fading between different links are dependent and they use some correlation matrix to model the correlations between fading.Then the capacity of point-to-point MIMO transmission is derived under dependent fading.
%Evaluating the interference correlation in wireless network is complicated: First, the spatial locations of interfering transmitters are uncertain, especially in ad hoc networks or heterogeneous networks with randomly scattered small cells.
%Second, closed-form results cannot be derived in the traditional grid model, e.g., hexagonal grid model for the cellular wireless network.
%Third, the fading and shadowing in the physical channel make the quantitative analysis difficult.
Recently, the tools from stochastic geometry \cite{stoyan1987stochastic} have been used extensively in modeling and analysis of wireless communication systems; see, e.g., \cite{haenggi2009stochastic, haenggi2009interference, haenggi2012stochastic, baccelli2009stochastic } and references therein.
This mathematical framework permits the derivation of closed-form results for various system metrics and makes it possible to evaluate the interference correlation.
A number of works considering the related problems are as follows.
In \cite{ganti2009spatial} the authors evaluated the spatio-temporal correlation coefficient of the interference and the joint probability of success in ALOHA networks, and in \cite{schilcher2011temporal} the authors calculated the correlation coefficient of interference under different assumptions of dependence.
%In our analysis, we apply the stochastic geometry tools and model the spatial distributions of transmitters as Poisson point process (PPP).
The framework for the analysis of the local delay was provided in \cite{net:Haenggi12tit}\cite{baccelli2010new}\cite{net:Haenggi10allerton}\cite{net:Haenggi10icc}, where different scenarios were considered and it was observed that the mean local delay may be infinite under certain system parameters.
The work in \cite{net:Gong13twc} extended the results to the case of finite mobility.
In \cite{gulati2012characterizing}, a new model, which characterizes different degrees of temporal dependence, was proposed to evaluate the local delay by using joint interference statistics.
In \cite{net:Zhang12twc}, the optimal power control policies for different fading statistics were proposed to minimize the mean local delay.
All the above works are based on the Poisson point process (PPP) model, while the work in \cite{alfano2011spatial} analyzed the local delay in clustered networks.

\subsection{Contributions}
In this work, we focus on the question that whether increasing randomness in the MAC helps reduce the local delay.
We apply the so-called Poisson bipolar model (see \cite[Sec. 5.3]{haenggi2012stochastic}), and derive the mean and variance of the local delay under FHMA and ALOHA.
Based on the mean and variance of the local delay we have derived, we explore the essential difference between the two MAC protocols.
We also evaluate the optimal number of sub-bands for FHMA and the optimal transmit probability for ALOHA that minimize the mean local delay.
The issue of optimizing the number of sub-bands was also considered in \cite{jindal2008bandwidth}, where the optimal number of sub-bands is derived to maximize the number of concurrent transmissions.
However, such outage-based framework used in \cite{jindal2008bandwidth} cannot capture the effects of correlated interference.
In the last part of our work, we evaluate the mean delay-jitter tradeoff and the bounds on the tail probability of the local delay, both of which are critical issues for the system design.

%Our results show that while the mean local delay of the two protocols, FHMA and ALOHA, are the same for certain parameters, the variances are rather different.
Our results reveal that the means of the local delay of the two protocols, FHMA and ALOHA, coincide when the number of sub-bands $N$ in FHMA is equal to the reciprocal of the transmit probability $p$ in ALOHA (with thermal noise ignored).
However, the variances of the local delay for the two protocols are drastically different: when $p=\frac{1}{N}$ and $N\rightarrow\infty$, the variance in FHMA converges to a constant which is typically small, while in ALOHA the variance scales as $\Theta(N^2)$.
Moreover, we calculate bounds on the complementary cumulative distribution function (ccdf) of the local delay when no MAC dynamic is introduced. In that case, the distribution of the local delay has a heavy tail, which results in an infinite mean local delay.
By employing the MAC randomness of either FHMA or ALOHA, the ccdf of the local delay will decay fast, and the mean local delay will then be finite.
This observation reveals the underlying mechanism why even such simple MAC protocols can greatly reduce the local delay.

The remaining part of this paper is organized as follows. Section \ref{sec:model} describes the network model and the MAC protocols. Section \ref{sec:mean_var} then establishes the main analytical results of this paper, including the mean and variance of the local delay for FHMA and for ALOHA. Section \ref{sec:opt_par} evaluates the optimal number of sub-bands for FHMA and the optimal transmit probability for ALOHA that minimize the mean local delay. Section \ref{sec:opt_thd} evaluates the optimal SINR threshold that minimizes the mean local delay. Section \ref{sec:design} presents the mean delay-jitter tradeoff and the bounds on the tail probability of the local delay, and Section \ref{sec:conclusion} offers the concluding remarks.

\section{System Model}
\label{sec:model}
\subsection{Network Model}
To obtain the most essential features, we consider the widely used Poisson bipolar model.
In this model, the locations of the transmitters are modeled as a PPP $\Phi=\{x_i\}\subset\mathbb{R}^d$ of intensity $\lambda$.
Each transmitter is associated with one receiver which is at a fixed distance $r_0$ to the corresponding transmitter.
In the analysis, we will condition on a particular desired transmitter $x_0\in\Phi$, and denote by $r_0=|x_0|$ the distance from this transmitter to the origin where the receiver resides.
Such conditioning is equivalent to adding the point $x_0$ to the PPP and guarantees that the link between $x_0$ and the origin is a typical link, in the sense that this link behaves statistically the same as all other links (see \cite[Ch. 8]{haenggi2012stochastic}).

\begin{figure}[!ht]
\centering
\includegraphics[width=0.5\textwidth]{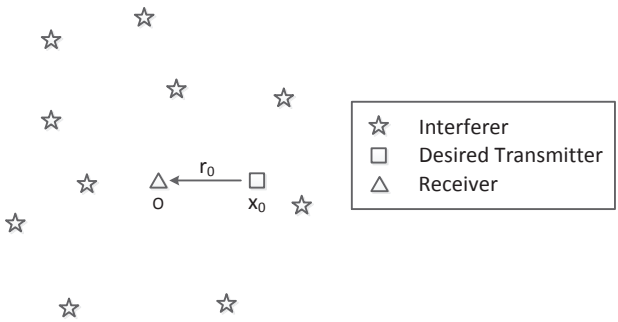}
\caption{Spatial distribution of different network entities.}
\label{fig:model1}
\end{figure}

We assume that the time is divided into discrete slots with equal duration.
Each transmission attempt occupies one time slot, and if a transmission fails in a certain time slot, a retransmission will be conducted.
The local delay is defined as the number of time slots until a packet is successfully received \cite{net:Haenggi12tit}\cite{baccelli2010new}.
In this paper, we assume fully backlogged nodes so that whenever a node is scheduled to access the channel it always has data to transmit. The local delay is thus basically the transmission delay, but not the queueing delay.

%For the propagation model, we assume that the power fading coefficients are spatially and temporally independent with exponential distribution of unit mean (i.e., Rayleigh fading). Let $\alpha$ be the path loss exponent and $h_{k,x}$ be the fading coefficient between transmitter $x$ and the considered receiver located at origin $o$ in time slot $k$.
%We assume that all transmitters transmit at a normalized power level of unity.
%This constant power assumption is consistent with the bipolar network model, in which all link distances are
%identical, and we set the constant power level to unity without loss of generality.
%The thermal noise is assumed to be white Gaussian with normalized power spectral density $N_0$.
For the propagation model, we consider the common path loss $l(r)=\kappa r^{-\alpha}$, where $\alpha$ is the path loss exponent and $\kappa$ is a constant.
We will further discuss the effect of bounded path loss $l(r)=\kappa (r^{\alpha}+\varepsilon)^{-1}$ in the subsection \ref{sec:bounded}.
We assume that the power fading coefficients are spatially and temporally independent with exponential distribution of unit mean (i.e., Rayleigh fading), and let $h_{k,x}$ be the fading coefficient between transmitter $x$ and the considered receiver located at origin $o$ in time slot $k$.
Without loss of generality, we assume that all transmitters transmit at a normalized power level of unity.
This constant power assumption is consistent with the bipolar network model, in which all link distances are
identical.
The thermal noise is assumed to be white Gaussian with power spectral density $N_r$.
To simplify the notations, we introduce the normalized noise power spectral density as $N_0=N_r/\kappa$.

We assume that the SINR threshold model is applied.
That is, for each time-frequency resource block, as long as the SINR is above a threshold $\theta$, it can be successfully used for information transmission at spectral efficiency $\log_2(1+\theta)$ bits per second per Hz.
We also assume that a packet of a fixed size needs exactly one time slot to be transmitted if it is allocated the entire frequency band $W$ under SINR threshold $\theta$ and successfully transmitted in that time slot.
In that way, in the FHMA case if the entire frequency band is split into $N$ sub-bands, a packet will need $N$ successful time slots.
%\footnote{Here, the value $N/\log_2(1+\theta)$ may not be an integer. In this case, the time slots needed for transmitting a packet successfully is $\lceil N/\log_2(1+\theta)\rceil$.
%To facilitate the notation, we assume that $N/\log_2(1+\theta)$ to be an integer in the following analysis. For the non-integer case, we only need to replace $N/\log_2(1+\theta)$ with $\lceil N/\log_2(1+\theta)\rceil$.}.
%\footnote{Here, the value $N/\log_2(1+\theta)$ may not be an integer. In this case, the analysis is still reasonable because it is possible that part of a slot is used for the transmission of current packet and the other part is used for the transmission of other packets. The delay evaluated in our work is just the time slots needed when transmitting current packet and the part of slots that are occupied by other packets are ignored.} since the rate is changed from $W$ to $\frac{W}{N}\log_2(1+\theta)$.
Meanwhile, in the ALOHA case, each active transmission will make use of the entire frequency band; thus, only one successful time slot is needed.
Notice that the local delay is measured by the number of time slots.
Since different system configurations may apply different durations of time slot,
we should normalize the local delay so that the actual delays of different system configurations can be compared fairly.
% but not the duration of time slot.
%In order to characterize the actual delay, we should take the duration of each time slot into consideration.
%Thus, in order to compare the actual delays of the systems whose durations of time slot may be different from each other, we should normalize the local delay.
The duration of each time slot is proportional to $\frac{1}{\log_2(1+\theta)}$ because the size of a packet is fixed and the spectral efficiency is proportional to $\log_2(1+\theta)$.
Therefore, when comparing the actual delays under different SINR thresholds $\theta$, we normalize the local delay by $\frac{1}{\log_2(1+\theta)}$ as the metric.
%This will be described in more detail later.

In static or moderately mobile network, the locations of the transmitters during all time slots are deemed to be correlated, resulting in the temporal interference correlation.
This type of correlation decreases the successful probability for retransmissions if the first transmission attempt failed, thus increasing the local delay.
In order to reduce the effect of interference correlation, we study two kinds of MAC randomness described as follows.

\subsection{FHMA}
In the FHMA case, we assume that the total frequency band $W$ is divided into $N$ sub-bands and each transmitter chooses a sub-band uniformly randomly, independently of the location and the time slot (i.e., memoryless both spatially and temporally).
Let $s\in\mathbb{S}=\{1,2,\cdots,N\}$ be the sub-band index, and let $\mathcal{S}_k(x)\in\mathbb{S}$ denote the index of the sub-band used by node $x\in\Phi$ in time slot $k$.
With these notations, the interference at the typical receiver located at the origin $o$ in time slot $k$ is given by
\begin{equation}
I_k=\sum_{x\in\Phi\backslash\{x_0\}}h_{k,x}\kappa|x|^{-\alpha}\mathbf{1}(\mathcal{S}_k(x)=\mathcal{S}_k(x_0)),
\end{equation}
where $\mathbf{1}(\cdot)$ is the indicator function and $|x|$ denotes the distance between $x$ and the origin $o$.
Note that the exclusion of $x_0$ from the sum over the point process does not imply that $x_0\notin\Phi$, but it ensures that when we condition on $x_0\in\Phi$, the power received from this node is not counted as interference.

Besides reducing the interference and breaking the correlation, introducing FHMA has the additional benefit that the noise power decreases from $W\kappa N_0$ to $\frac{W}{N}\kappa N_0$. By taking this noise scaling into consideration,
we obtain the SINR of the typical receiver in time slot $k$ as
\begin{eqnarray}
\mathrm{SINR}_k=\qquad\qquad\qquad\qquad\qquad\qquad\quad\qquad\nonumber\\
\frac{h_{k,x_0}r_0^{-\alpha}}{\frac{WN_0}{N}+\sum_{x\in\Phi\backslash\{x_0\}}h_{k,x}|x|^{-\alpha}\mathbf{1}(\mathcal{S}_k(x)=\mathcal{S}_k(x_0))}.
\end{eqnarray}
\subsection{ALOHA}
In the ALOHA case, let $\Phi_k$ be the transmitting set in time slot $k$.
The interference at the typical receiver located at the origin $o$ in time slot $k$ is
\begin{equation}
I_k=\sum_{x\in\Phi\backslash\{x_0\}}h_{k,x}\kappa|x|^{-\alpha}\mathbf{1}(x\in\Phi_k).
\end{equation}

Unlike FHMA, the noise scaling effect does not exist for ALOHA since the entire frequency band is used for each transmission.
The SINR of the typical receiver in time slot $k$ is
\begin{equation}
\mathrm{SINR}_k=\frac{h_{k,x_0}r_0^{-\alpha}}{WN_0+\sum_{x\in\Phi\backslash\{x_0\}}h_{k,x}|x|^{-\alpha}\mathbf{1}(x\in\Phi_k)}.
\end{equation}

%A conjecture is that the division of frequency band will diminish the interference correlation and reduce the local delay.
%In the following, we will focus on evaluating the influence of frequency splitting and figure out whether the best number of sub-bands $N$ exists to minimize the local delay.

\section{Mean and Variance of the Local Delay}
\label{sec:mean_var}
In this section, we derive the mean and variance of the local delay for FHMA and for ALOHA respectively.
%We only describe the detail of derivations in the FHMA case, while the derivations in ALOHA case is similar and will be described briefly.
\subsection{FHMA}
\subsubsection{Mean local delay}
The following theorem gives the mean local delay in FHMA networks.
\begin{thm}
\label{thm:ld fh}
In FHMA with $N$ sub-bands, the mean local delay is
\begin{eqnarray}
%D(N)&=&\begin{cases}
%\infty&N=1\\
%N\exp\left(\frac{A}{(N-1)^{1-\delta}N^{\delta}}+\frac{B}{N}\right)&N>1,
%\end{cases}
D(N)&=&N\exp\left(\frac{A}{(N-1)^{1-\delta}N^{\delta}}+\frac{B}{N}\right),\label{equ:localdelay}
\end{eqnarray}
where $A=\lambda c_d r_0^d\theta^{\delta}C(\delta)$, $B=\theta r_0^\alpha WN_0$, $\delta=d/\alpha$, $C(\delta)=1/{\mathrm{sinc}(\delta)}$, and $c_d=|b(o,1)|$ is the volume of the $d$-dimensional unit ball\footnote{Since the equation (\ref{equ:localdelay}) has implied that $D(1)=\infty$, without loss of generality, we can regard the domain of $D(N)$ as $N\geqslant1$ with $D(1)=\infty$.}.
\end{thm}
\begin{proof}
Let $\mathcal{C}_\Phi$ be the event that a transmission succeeds conditioned on the PPP $\Phi$. The probability for successful transmission given $\Phi$ is the same for each time slot. Our analysis below is conditioned on $\Phi$ having a point at $x_0$. This means that the probability measure of the point process is the Palm probability $\mathbb{P}^{x_0}$ (see Ch. 8 in \cite{haenggi2012stochastic}).
Correspondingly, the expectation, denoted by $\mathbb{E}^{x_0}$, is taken with respect to the measure $\mathbb{P}^{x_0}$.
With this notation, by setting the SINR threshold to be $\theta$, we denote the probability of successful transmission conditioned on $\Phi$ as $\mathbb{P}^{x_0}(\mathcal{C}_\Phi)=\mathbb{P}^{x_0}(\mathrm{SINR}_k>\theta\mid\Phi)$, which can be evaluated as
\begin{eqnarray}
\!\!\!&\!\!\!\!\!\!&\mathbb{P}^{x_0}(\mathcal{C}_\Phi)=\mathbb{P}^{x_0}(\mathrm{SINR}_k>\theta\mid\Phi) \nonumber\\
\!\!\!&\!\!\!=\!\!\!&\mathbb{P}^{x_0}\Big(h_{k,x_0}r_0^{-\alpha}>\theta \Big(\frac{W}{N}N_0+I_k\Big)\mid\Phi\Big) \nonumber\\
\!\!\!&\!\!\!\stackrel{(a)}{=}\!\!\!&\mathbb{E}^{x_0}\Big(\exp\Big(-\theta r_0^\alpha \Big(\frac{W}{N}N_0+I_k\Big)\Big)\mid\Phi\Big) \nonumber\\
\!\!\!&\!\!\!=\!\!\!&\mathbb{E}^{x_0}\Big(\exp\Big(-\theta r_0^\alpha\frac{W}{N}N_0- \nonumber\\
\!\!\!&\!\!\!\!\!\!& \sum_{x\in\Phi\backslash\{x_0\}}\theta r_0^\alpha h_{k,x}|x|^{-\alpha}\mathbf{1}(\mathcal{S}_k(x)=\mathcal{S}_k(x_0))\Big)\mid\Phi\Big) \nonumber\\
\!\!\!&\!\!\!=\!\!\!&\exp\Big(-\frac{\theta r_0^\alpha WN_0}{N}\Big) \nonumber\\
\!\!\!&\!\!\!\!\!\!&\!\!\!\!\!\!\!\prod_{x\in\Phi\backslash\{x_0\}}\!\!\!\mathbb{E}^{x_0}\left(\exp\left(-\theta r_0^\alpha h_{k,x}|x|^{-\alpha}\mathbf{1}(\mathcal{S}_k(x)=\mathcal{S}_k(x_0))\right)\mid\Phi\right) \nonumber\\
\!\!\!&\!\!\!=\!\!\!&\exp\Big(-\frac{\theta r_0^\alpha WN_0}{N}\Big) \nonumber\\
\!\!\!&\!\!\!\!\!\!&\!\!\!\!\prod_{x\in\Phi\backslash\{x_0\}}\!\!\!\!\Big(\frac{1}{N}\mathbb{E}^{x_0}\left(\exp\left(-\theta r_0^\alpha h_{k,x}|x|^{-\alpha}\right)\mid\Phi\right)+\frac{N-1}{N}\Big) \nonumber\\
\!\!\!&\!\!\!\stackrel{(b)}{=}\!\!\!&\exp\Big(-\frac{\theta r_0^\alpha WN_0}{N}\Big)\!\!\!\!\!\prod_{x\in\Phi\backslash\{x_0\}}\!\!\!\!\!\Big(\frac{1}{N}\frac{1}{1+\theta r_0^\alpha|x|^{-\alpha}}+\frac{N-1}{N}\Big). \nonumber\\
\!\!\!&& \label{equ:succ}
\end{eqnarray}

In steps $(a)$ and $(b)$ of the derivation above, we have applied the property that the fading coefficients $h_{k,x}$ are i.i.d. random variables with exponential distribution of unit mean.
The number of time slots needed until a successful time slot appears, denoted by $\Delta$, is a random variable called \emph{delay till success} (DTS) \cite{net:Zhang12twc}.
%because the event for successful transmitting is a random event.
Conditioned upon $\Phi$, the success events in different time slots are independent with probability $\mathbb{P}^{x_0}(\mathcal{C}_\Phi)$; therefore, the DTS with given $\Phi$, denoted by $\Delta_\Phi$, is a random variable with geometric distribution given by
\begin{equation}
\mathbb{P}^{x_0}\left(\Delta_\Phi=k\right)=\left(1-\mathbb{P}^{x_0}(\mathcal{C}_\Phi)\right)^{k-1}\mathbb{P}^{x_0}(\mathcal{C}_\Phi). \label{equ:pdfdelta}
\end{equation}
The conditional expectation of $\Delta_\Phi$ is taken w.r.t. the fading and the MAC, given by $\mathbb{E}^{x_0}\left(\Delta_\Phi\right)=1/\mathbb{P}^{x_0}(\mathcal{C}_\Phi)$.
Noticing that a packet will need $N$ successful time slots to finish transmission in FHMA, the mean local delay can be evaluated as
\begin{eqnarray}
D(N)&\!\!\!=\!\!\!&N\mathbb{E}^{x_0}\left(\Delta\right) \nonumber \\
&\!\!\!=\!\!\!&N\mathbb{E}^{x_0}_\Phi\left(\mathbb{E}^{x_0}\left(\Delta_\Phi\right)\right) \nonumber \\
%\end{eqnarray}
%
%The local delay can be obtained by taking the expectation over $\Phi$ as follows
%\begin{eqnarray}
&\!\!\!=\!\!\!&N\mathbb{E}^{x_0}_\Phi\Big(\frac{1}{\mathbb{P}^{x_0}(\mathcal{C}_\Phi)}\Big) \nonumber \\
&\!\!\!\stackrel{(a)}{=}\!\!\!&N\exp\Big(\frac{\theta r_0^\alpha WN_0}{N}\Big) \nonumber \\
&\!\!\!\!\!\!&\mathbb{E}^{x_0}_\Phi\bigg(\frac{1}{\prod_{x\in\Phi\backslash\{x_0\}}\left(\frac{1}{N}\frac{1}{1+\theta r_0^\alpha|x|^{-\alpha}}+\frac{N-1}{N}\right)}\bigg) \nonumber \\
&\!\!\!=\!\!\!&N\exp\left(\frac{\theta r_0^\alpha WN_0}{N}\right) \nonumber \\
&\!\!\!\!\!\!&\mathbb{E}^{x_0}_\Phi\bigg(\prod_{x\in\Phi\backslash\{x_0\}}\frac{1}{\frac{1}{N}\frac{1}{1+\theta r_0^\alpha|x|^{-\alpha}}+\frac{N-1}{N}}\bigg).
\end{eqnarray}
where $(a)$ follows from (\ref{equ:succ}).
By applying the probability generating functional (PGFL) of the PPP, we obtain
\begin{eqnarray}
\!\!\!\!&\!\!\!\!\!\!&D(N) \nonumber \\
\!\!\!\!&\!\!\!=\!\!\!&N\exp\left(\frac{\theta r_0^\alpha WN_0}{N}\right) \nonumber \\
\!\!\!\!&\!\!\!\!\!\!&\exp\bigg(-\lambda\int_{\mathbb{R}^d}\bigg(1-\frac{1}{\frac{1}{N}\frac{1}{1+\theta r_0^\alpha|x|^{-\alpha}}+\frac{N-1}{N}}\bigg)\mathrm{d}x\bigg) \nonumber \\
\!\!\!\!&\!\!\!=\!\!\!&N\exp\bigg(\frac{\theta r_0^\alpha WN_0}{N}+\lambda c_dd\int_0^\infty \frac{r^{d-1}}{\frac{N}{\theta r_0^\alpha}r^{\alpha}+N-1}\mathrm{d}r\bigg) \nonumber \\
%&=&\frac{N}{\log_2(1+\theta)}\exp\left((N-1)^{\frac{1}{\alpha}-1}N^{-\frac{1}{\alpha}}2\pi\lambda r\theta^{\frac{1}{\alpha}}\frac{1}{\alpha}\Gamma\left(\frac{1}{\alpha}\right)\Gamma\left(1-\frac{1}{\alpha}\right)\right) \nonumber \\
\!\!\!\!&\!\!\!=\!\!\!&N\exp\bigg(\frac{\lambda c_d r_0^d\theta^{\delta}C(\delta)}{(N-1)^{1-\delta}N^{\delta}}+\frac{\theta r_0^\alpha WN_0}{N}\bigg).
%\begin{cases}
%\infty&N=1\\
%N\exp\left(\frac{\lambda c_d r_0^d\theta^{\delta}C(\delta)}{(N-1)^{1-\delta}N^{\delta}}+\frac{\theta r_0^\alpha WN_0}{N}\right)&N>1,
%\end{cases}
\end{eqnarray}
where $\delta=d/\alpha$, $C(\delta)$ is given by $C(\delta)=\Gamma\left(1+\delta\right)\Gamma\left(1-\delta\right)=\frac{1}{\mathrm{sinc}(\delta)}$, and $c_d=|b(o,1)|$ is the volume of the $d$-dimensional unit ball.
\end{proof}

The result in Theorem \ref{thm:ld fh} is closed-form and easy to evaluate and interpret.
%Letting $A=\lambda c_d r_0^d\theta^{\delta}C(\delta)$ and $B=\theta r_0^\alpha WN_0$, the local delay can be written as
%\begin{eqnarray}
%D(N)&=&\begin{cases}
%\infty&N=1\\
%\frac{N}{\log_2(1+\theta)}\exp\left(\frac{A}{(N-1)^{1-\delta}N^{\delta}}+\frac{B}{N}\right)&N>1,
%\end{cases}
%\end{eqnarray}
The value of $A$ is determined by the interference and that of $B$ is due to the thermal noise.
%For finite $N$, the benefit of increasing $N$ on the interference part is slightly smaller than that on the noise part.
%When $N$ is large, we have $\frac{N-1}{N}\sim1$, and the mean local delay can be approximated by
From (\ref{equ:localdelay}), we have
\begin{eqnarray}
\!\!\!\!&&\!\!\!\!\!\!\!\!D(N)=N\exp\bigg(A\bigg(1-\frac{1}{N}\bigg)^{\delta-1}\frac{1}{N}+\frac{B}{N}\bigg) \nonumber \\
\!\!\!\!&\!\!\!=\!\!\!&N\exp\bigg(A\bigg(1-\frac{\delta-1}{N}+O\bigg(\frac{1}{N^2}\bigg)\bigg)\frac{1}{N}+\frac{B}{N}\bigg) \nonumber \\
\!\!\!\!&\!\!\!=\!\!\!&N\exp\bigg(\frac{A+B}{N}+O\bigg(\frac{1}{N^2}\bigg)\bigg) \nonumber \\
\!\!\!\!&\!\!\!=\!\!\!&N+A+B+O\bigg(\frac{1}{N}\bigg).
%D(N)&\sim&N\exp\left(\frac{A+B}{N}\right)\quad\quad N\rightarrow\infty \nonumber \\
%&\sim&N\quad\quad N\rightarrow\infty.
\end{eqnarray}
The result shows that when $N$ is large, the mean local delay increases linearly with $N$.
Since $D(1)$ is infinity, there exists an optimal number of sub-bands $N_{\mathrm{opt}}$ that minimizes the mean local delay.
Inspecting $D(N)$, we see that there are two effects by splitting the entire frequency band into $N$ sub-bands:
first, the mean local delay $D(N)$ tends to decrease due to the reduced interference correlation;
second, $D(N)$ tends to increase since the number of time slots needed becomes $N$ times larger.
%first, the mean local delay $D(N)$ tends to increase since the number of time slots needed becomes $N$ times larger;
%second, $D(N)$ tends to decrease due to the smaller correlation.
In view of this, we introduce two regimes, \emph{correlation-limited} regime and \emph{bandwidth-limited} regime.
For $N<N_{\mathrm{opt}}$, the first effect outweighs the second one, and the network operates in the correlation-limited regime.
For $N>N_{\mathrm{opt}}$, it is the opposite and the network operates in the bandwidth-limited regime.

In the above, we have derived results under the assumption that the frequency allocation is dynamic (i.e., the sub-bands are allocated randomly and independently in each time slot).
Alternatively, one could consider the case where the frequency allocation is static over time.
%There is another case in which the frequency allocation is also static during all the time slots.
That case is exactly the same as the case where no frequency splitting is applied, with the only difference that the intensity of the interfering transmitters is scaled down to $\lambda/N$. The mean local delay in that case is also infinite.
This fact explains that even though frequency splitting is introduced, if the sub-bands are not reallocated randomly temporally,
the mean local delay will still be infinite. This is a nontrivial observation since it reveals that the reduction of the mean local delay by introducing FHMA does not come from reducing the interference or the thermal noise, but mainly comes from reducing the interference correlation.

Based on Theorem \ref{thm:ld fh}, we show how the normalized mean local delay $\frac{D(N)}{\log_2(1+\theta)}$ varies with $N$ numerically. As for the parameters, we ignore the thermal noise ($N_0=0$) and set the intensity of transmitters as $\lambda=0.01\mathrm{m}^{-2}$ by default, which means that the coverage area of each transmitter is $100\mathrm{m}^2$ on average, reasonable for a typical deployment of WLAN. The path loss exponent is set as $\alpha=4$ by default, and the distance between the receiver and the typical desired transmitter is $r_0=5\mathrm{m}$. Let $\theta$ be the outage threshold for SINR. The relationship between $\frac{D(N)}{\log_2(1+\theta)}$ and $N$ is depicted in Fig. \ref{fig:change}.

By changing the values of $\alpha$ and $\lambda$ respectively, we get the curves in Fig. \ref{fig:change}.
Comparing the curves in Fig. \ref{fig:default} with those in Fig. \ref{fig:change:a} and Fig. \ref{fig:change:b}, we observe that the optimal number of sub-bands increases when $\alpha$ decreases or when $\lambda$ increases. This observation is consistent with the intuition: Smaller $\alpha$ implies that the signal strength decays more slowly with distance, and larger $\lambda$ implies that more transmitters exist in the same region, so in both cases more interference is created. Therefore, the entire frequency band should be divided into more sub-bands, namely larger $N_{\mathrm{opt}}$, to reduce the interference and interference correlation.

\begin{figure}
  \centering
  \subfigure[$\lambda=0.01$ and $\alpha=4$.]{
    \label{fig:default} %% label for second subfigure
    \includegraphics[width=0.45\textwidth]{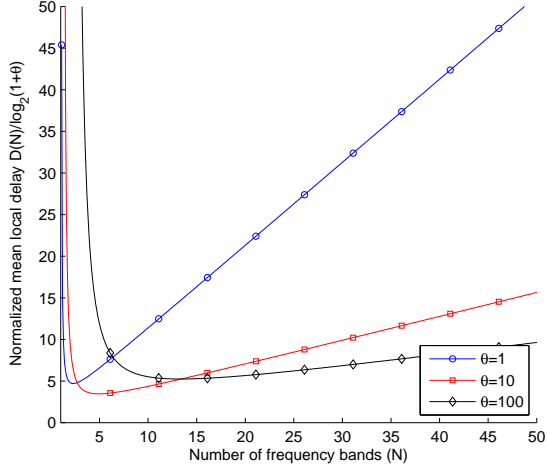}}
  \subfigure[$\lambda=0.01$ and $\alpha=3$.]{
    \label{fig:change:a} %% label for first subfigure
    \includegraphics[width=0.45\textwidth]{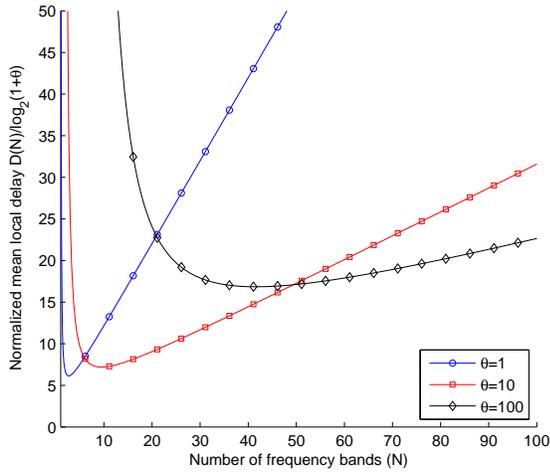}}
%  \hspace{0.2in}
  \subfigure[$\lambda=0.04$ and $\alpha=4$.]{
    \label{fig:change:b} %% label for second subfigure
    \includegraphics[width=0.45\textwidth]{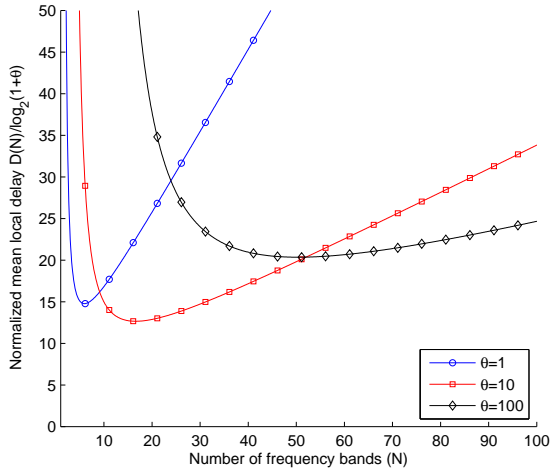}}
  \caption{The normalized mean local delay $\frac{D(N)}{\log_2(1+\theta)}$ as a function of the number of sub-bands $N$, when $d=2$, $r=5$m, and thermal noise ignored.}
  \label{fig:change} %% label for entire figure
\end{figure}

\subsubsection{Variance of the local delay}
The mean local delay discussed above has characterized the mean number of time slots needed until a packet is successfully transmitted.
In order to better understand the distribution of the local delay, we also derive its variance.
The following theorem gives the variance of the local delay for FHMA.
\begin{thm}
\label{thm:var fp}
In FHMA with $N$ sub-bands, the variance of the local delay is
\begin{equation}
V(N)=N\left(N+1\right)\exp\bigg(\frac{(2N-1-\delta)A}{N^{\delta}(N-1)^{2-\delta}}+\frac{2B}{N}\bigg)\nonumber\\
\end{equation}
\begin{equation}
\qquad\qquad-D(N)-D^2(N). \label{equ:var}
%\begin{cases}
%\infty&N=1\\
%N\left(N+1\right)\exp\Big(\frac{(2N-1-\delta)A}{N^{\delta}(N-1)^{2-\delta}}+\frac{2B}{N}\Big)-D(N)-D^2(N)&N>1
%\end{cases}  \label{equ:var}
\end{equation}
\end{thm}
\begin{proof}
In order to transmit a packet in FHMA, $N$ successful transmissions are needed.
Letting $\Delta_i$ $(1\leq i\leq N)$ be the DTS of the $i$th transmission, we get the local delay of a packet as $\sum_{i=1}^{N}\Delta_i$.
For $1\leq i,j\leq N$ and $i\neq j$, $\Delta_i$ and $\Delta_j$ are dependent because the interference of the $i$th transmission and that of the $j$th transmission are correlated.
However, if we condition on $\Phi$, $\{\Delta_i\}$ are i.i.d. random variables with geometric distribution given by (\ref{equ:pdfdelta}).
With these notations, we obtain the variance of the local delay as
%Notice that $D(N)=\frac{N}{\log_2(1+\theta)}\mathbb{E}^{x_0}\left(\Delta\right)$, we get the variance of the delay as follows
\begin{eqnarray}
\!\!\!&\!\!\!\!\!\!&\!\!\!\!\!\!V(N)=\mathbb{E}^{x_0}\left(\left(\sum_{i=1}^{N}\Delta_i\right)^2\right)-\left(\mathbb{E}^{x_0}\left(\sum_{i=1}^{N}\Delta_i\right)\right)^2 \nonumber \\
\!\!\!&\!\!\!=\!\!\!&\mathbb{E}^{x_0}\left(\sum_{i=1}^{N}\Delta_i^2+\sum_{\stackrel{i,j=1}{i\neq j}}^{N}2\Delta_i\Delta_j\right)-\left(\sum_{i=1}^{N}\mathbb{E}^{x_0}\left(\Delta_i\right)\right)^2 \nonumber \\
\!\!\!&\!\!\!\stackrel{(a)}{=}\!\!\!&\sum_{i=1}^{N}\mathbb{E}^{x_0}\left(\Delta_i^2\right)+\sum_{\stackrel{i,j=1}{i\neq j}}^{N}2\mathbb{E}^{x_0}\left(\Delta_i\Delta_j\right)-D^2(N), \nonumber
\end{eqnarray}
where $(a)$ follows from the definition of the mean local delay.
%Notice that conditioning on $\Phi$, $\{\Delta_i\}$ are i.i.d. random variables with geometric distribution given by (\ref{equ:pdfdelta}).
By applying the total expectation formula, we have
\begin{eqnarray}
\!\!\!V(N)\!\!\!&\!\!\!=\!\!\!&\!\!\!\sum_{i=1}^{N}\mathbb{E}^{x_0}_\Phi\left(\mathbb{E}^{x_0}\left(\Delta_i^2\mid\Phi\right)\right) \nonumber \\
&&+\sum_{\stackrel{i,j=1}{i\neq j}}^{N}2\mathbb{E}^{x_0}_\Phi\left(\mathbb{E}^{x_0}\left(\Delta_i\Delta_j\mid\Phi\right)\right)-D^2(N) \nonumber \\
\!\!\!&\stackrel{(b)}{=}&\!\!\!N\mathbb{E}^{x_0}_\Phi \left(\frac{2-\mathbb{P}^{x_0}(\mathcal{C}_\Phi)}{(\mathbb{P}^{x_0}(\mathcal{C}_\Phi))^2}\right)  \nonumber\\
\!\!\!&&+N\left(N-1\right)\mathbb{E}^{x_0}_\Phi \left(\frac{1}{(\mathbb{P}^{x_0}(\mathcal{C}_\Phi))^2}\right)-D^2(N) \nonumber\\
\!\!\!&\!\!\!=\!\!\!&\!\!\!N\left(N+1\right)\mathbb{E}^{x_0}_\Phi\left(\frac{1}{(\mathbb{P}^{x_0}(\mathcal{C}_\Phi))^2}\right) \nonumber\\
\!\!\!&&-N\mathbb{E}^{x_0}_\Phi\left(\frac{1}{\mathbb{P}^{x_0}(\mathcal{C}_\Phi)}\right)-D^2(N)\nonumber\\
\!\!\!&\!\!\!=\!\!\!&\!\!\!N\left(N+1\right)\mathbb{E}^{x_0}_\Phi\left(\frac{1}{(\mathbb{P}^{x_0}(\mathcal{C}_\Phi))^2}\right)\nonumber \\
\!\!\!&&-D(N)-D^2(N),  \label{equ:vartmp}
%V(N)&=&\frac{N^2}{(\log_2(1+\theta))^2}\left(\mathbb{E}^{x_0}\left(\Delta^2\right)-\left(\mathbb{E}^{x_0}\left(\Delta\right)\right)^2\right) \nonumber \\
%&=&\frac{N^2}{(\log_2(1+\theta))^2}\mathbb{E}^{x_0}\left(\Delta^2\right)-D^2(N) \nonumber \\
%&=&\frac{N^2}{(\log_2(1+\theta))^2}\mathbb{E}^{x_0}_\Phi\left(\mathbb{E}^{x_0}\left(\Delta_\Phi^2\right)\right)-D^2(N) \nonumber \\
%&\stackrel{(a)}{=}&\frac{N^2}{(\log_2(1+\theta))^2}\mathbb{E}^{x_0}_\Phi\left(\frac{2-\mathbb{P}^{x_0}(\mathcal{C}_\Phi)}{(\mathbb{P}^{x_0}(\mathcal{C}_\Phi))^2}\right)-D^2(N) \nonumber \\
%&=&\frac{2N^2}{(\log_2(1+\theta))^2}\mathbb{E}^{x_0}_\Phi\left(\frac{1}{(\mathbb{P}^{x_0}(\mathcal{C}_\Phi))^2}\right)-\frac{N}{\log_2(1+\theta)}D(N)-D^2(N), \label{equ:vartmp}
\end{eqnarray}
where $(b)$ follows from the second moment of the geometrically distributed random variable.

From (\ref{equ:succ}), we have
\begin{eqnarray}
%&&\mathbb{E}^{x_0}_\Phi\left(\frac{1}{(\mathbb{P}^{x_0}(\mathcal{C}_\Phi))^2}\right)\nonumber \\
\!\!\!\!\!\!\!\!\!&\!\!\!\!\!\!&\mathbb{E}^{x_0}_\Phi\left(\frac{1}{(\mathbb{P}^{x_0}(\mathcal{C}_\Phi))^2}\right) \nonumber \\
\!\!\!\!\!\!\!\!\!&\!\!\!=\!\!\!&\exp\left(\frac{2\theta r_0^\alpha WN_0}{N}\right)\nonumber\\
\!\!\!\!\!\!&&\quad\mathbb{E}^{x_0}_\Phi\Bigg(\frac{1}{\prod_{x\in\Phi\backslash\{x_0\}}\left(\frac{1}{N}\frac{1}{1+\theta r_0^\alpha|x|^{-\alpha}}+\frac{N-1}{N}\right)^2}\Bigg) \nonumber \\
\!\!\!\!\!\!\!\!\!&\!\!\!\stackrel{(c)}{=}\!\!\!&\exp\Bigg(\frac{2\theta r_0^\alpha WN_0}{N}\nonumber\\
\!\!\!\!\!\!&&\quad-\lambda\int_{\mathbb{R}^d}\Bigg(1-\frac{1}{\left(\frac{1}{N}\frac{1}{1+\theta r_0^\alpha|x|^{-\alpha}}+\frac{N-1}{N}\right)^2}\Bigg)\mathrm{d}x\Bigg) \nonumber \\
\!\!\!\!\!\!\!\!\!&\!\!\!=\!\!\!&\exp\Big(\frac{2\theta r_0^\alpha WN_0}{N}\nonumber\\
\!\!\!\!\!\!&&\!\!\!-\lambda c_dd\int_0^\infty\left(1-\frac{N^2(1+\theta r_0^\alpha r^{-\alpha})^2}{\left(N+(N-1)\theta r_0^\alpha r^{-\alpha}\right)^2}\right)r^{d-1}\mathrm{d}r\Big) \nonumber \\
\!\!\!\!\!\!\!\!\!&\!\!\!=\!\!\!&\exp\left(\frac{2\theta r_0^\alpha WN_0}{N}+\frac{\lambda c_d r_0^d\theta^\delta C(\delta)(2N-1-\delta)}{N^{\delta}(N-1)^{2-\delta}}\right),\nonumber\\
\!\!\!\!\!\!&& \label{equ:quadprob}
%\begin{cases}
%\infty&N=1\\
%\exp\left(\frac{2\theta r_0^\alpha WN_0}{N}+\frac{\lambda c_d r_0^d\theta^\delta C(\delta)(2N-1-\delta)}{N^{\delta}(N-1)^{2-\delta}}\right)&N>1,
%\end{cases} \label{equ:quadprob}
%&=&\exp\left(2\theta r_0^\alpha WN_0N^{-1}+\lambda c_d r_0^d\theta^\delta C(\delta)(2N-1-\delta)N^{-\delta}(N-1)^{\delta-2}\right),
\end{eqnarray}
where $(c)$ follows by applying the PGFL of the PPP. Plugging (\ref{equ:quadprob}) into (\ref{equ:vartmp}), we get the variance of the local delay as in Theorem \ref{thm:var fp}.
\end{proof}

\subsection{ALOHA}
%It is valuable to compare the FHMA scheme considered in our work to the ALOHA scheme.
The fundamental difference between FHMA and ALOHA is that if a packet is to be transmitted during a time slot, in FHMA the packet will be surely transmitted by randomly choosing a sub-band, while in ALOHA the packet will only be transmitted with a given probability.
Similar to the analysis of FHMA, we also assume that a packet needs exactly one time slot if it is allocated the entire frequency band $W$ under SINR threshold $\theta$ and successfully transmitted in that time slot.
We assume that each node transmits with probability $p$ in each time slot and if it transmits, it will make use of the entire frequency band.
In that way, only one successful time slot is needed to transmit a packet, and the local delay is the DTS of one transmission, denoted by $\Delta$.
%Next, we will derive the local delay and the delay variance respectively for ALOHA scheme; then we will compare the results of ALOHA scheme when $p=\frac{1}{N}$ with that of FHMA scheme when the entire frequency band is divided into $N$ sub-bands.
\subsubsection{Mean local delay}
The following theorem gives the mean local delay for ALOHA.
\begin{thm}
\label{thm:ld aloha}
In ALOHA with transmit probability $p$, the mean local delay is
\begin{eqnarray}
\widetilde{D}(p)&=&\frac{1}{p}\exp\bigg(\frac{pA}{(1-p)^{1-\delta}}+B\bigg).\label{equ:localdelay_aloha}
%\begin{cases}
%\infty&N=1\\
%\frac{1}{p}\exp\left(\frac{pA}{(1-p)^{1-\delta}}+B\right)&N>1
%\end{cases}\label{equ:localdelay_aloha}
\end{eqnarray}
\end{thm}
\begin{proof}
In each time slot, a packet will be transmitted with probability $p$ and the transmission will be successful with probability $\mathbb{P}^{x_0}(\mathrm{SINR}_k>\theta\mid\Phi)$ conditioned upon $\Phi$. Therefore, similar to the derivation of (\ref{equ:succ}), the probability for successfully transmitting a packet conditioned upon $\Phi$ in a time slot is
\begin{eqnarray}
\!\!\!\!\!&\!\!\!\!&\!\!\!\!\!\mathbb{P}^{x_0}(\mathcal{C}_\Phi)\stackrel{(a)}{=}p\mathbb{P}^{x_0}(\mathrm{SINR}_k>\theta\mid\Phi) \nonumber\\
\!\!\!&\!\!\!\!=\!\!\!\!&p\mathbb{P}^{x_0}\left(h_{k,x_0}r_0^{-\alpha}>\theta \left({W}N_0+I_k\right)\mid\Phi\right) \nonumber\\
\!\!\!&\!\!\!\!\stackrel{(b)}{=}\!\!\!\!&p\mathbb{E}^{x_0}\left(\exp\left(-\theta r_0^\alpha \left({W}N_0+I_k\right)\right)\mid\Phi\right) \nonumber\\
\!\!\!&\!\!\!\!=\!\!\!\!&p\mathbb{E}^{x_0}\Big(\exp\Big(-\theta r_0^\alpha{W}N_0 \nonumber\\
\!\!\!&\!\!\!\!&-\!\!\sum_{x\in\Phi\backslash\{x_0\}}\!\!\theta r_0^\alpha h_{k,x}|x|^{-\alpha}\mathbf{1}(x\in\Phi_k)\Big)\mid\Phi\Big) \nonumber\\
\!\!\!&\!\!\!\!=\!\!\!\!&p\exp\left(-{\theta r_0^\alpha WN_0}\right)\nonumber\\
\!\!\!&\!\!\!\!&\!\!\!\!\!\!\prod_{x\in\Phi\backslash\{x_0\}}\!\!\!\!\mathbb{E}^{x_0}\left(\exp\left(-\theta r_0^\alpha h_{k,x}|x|^{-\alpha}\mathbf{1}(x\in\Phi_k)\right)\mid\Phi\right) \nonumber\\
\!\!\!&\!\!\!\!=\!\!\!\!&p\exp\left(-{\theta r_0^\alpha WN_0}\right)\nonumber\\
\!\!\!&\!\!\!\!&\!\!\!\!\!\!\prod_{x\in\Phi\backslash\{x_0\}}\!\!\!\!\Big(p\mathbb{E}^{x_0}\Big(\exp\big(-\theta r_0^\alpha h_{k,x}|x|^{-\alpha}\big)\mid\Phi\Big)+1-p\Big) \nonumber\\
\!\!\!&\!\!\!\!\stackrel{(c)}{=}\!\!\!\!&p\exp\big(-\theta r_0^\alpha WN_0\big)\!\!\!\!\!\!\prod_{x\in\Phi\backslash\{x_0\}}\!\!\!\!\!\!\Big(\frac{p}{1+\theta r_0^\alpha|x|^{-\alpha}}+1-p\Big). \nonumber\\
&\!\!\!\!& \label{equ:succ_aloha}
\end{eqnarray}
where $(a)$ is because a transmission occurs with probability $p$, and $(b)$ and $(c)$ follows because the fading coefficients $h_{k,x}$ are i.i.d. random variables with exponential distribution of unit mean.
Then, the mean local delay for ALOHA is given by
\begin{eqnarray}
\widetilde{D}(p)&\!\!\!\!=\!\!\!\!&\mathbb{E}^{x_0}_\Phi\Big(\frac{1}{\mathbb{P}^{x_0}(\mathcal{C}_\Phi)}\Big) \nonumber \\
&\!\!\!\!=\!\!\!\!&\frac{1}{p}\exp\Big(\theta r_0^\alpha WN_0\nonumber\\
&&-\lambda\int_{\mathbb{R}^d}\Big(1-\frac{1}{\frac{p}{1+\theta r_0^\alpha|x|^{-\alpha}}+1-p}\Big)\mathrm{d}x\Big) \nonumber \\
&\!\!\!\!=\!\!\!\!&\frac{1}{p}\exp\left(\frac{\lambda c_d r_0^d\theta^{\delta}C(\delta)p}{(1-p)^{1-\delta}}+\theta r_0^\alpha WN_0\right).\nonumber\\
&& \label{variance_aloha_tmp}
\end{eqnarray}
Applying the definition of $A$ and $B$ in Theorem \ref{thm:ld fh}, we obtain the result in Theorem \ref{thm:ld aloha}.
\end{proof}

\subsubsection{Variance of the local delay}
The variance of the local delay in ALOHA is given by the following theorem.
\begin{thm}
\label{thm:var aloha}
In ALOHA with transmit probability $p$, the variance of the local delay is
\begin{eqnarray}
\widetilde{V}(p)&=&\frac{2}{p^2}\exp\Big(\frac{(2-p-\delta p)pA}{(1-p)^{2-\delta}}+2B\Big) \nonumber\\
&&\qquad-\widetilde{D}(p)-\widetilde{D}^2(p).\label{equ:var_aloha}
%\begin{cases}
%\infty&p=1\\
%\frac{2}{p^2}\exp\Big(\frac{(2-p-\delta p)pA}{(1-p)^{2-\delta}}+2B\Big)-\widetilde{D}(p)-\widetilde{D}^2(p)&0<p<1.
%\end{cases} \label{equ:var_aloha}
\end{eqnarray}
\end{thm}
\begin{proof}
In the ALOHA case, in order to transmit a packet, one successful transmission is needed.
The variance of local delay for ALOHA is thus
\begin{eqnarray}
\widetilde{V}(p)&\!\!\!\!=\!\!\!\!&\mathbb{E}^{x_0}\left(\Delta^2\right)-\left(\mathbb{E}^{x_0}\left(\Delta\right)\right)^2 \nonumber \\
&\!\!\!\!\stackrel{(a)}{=}\!\!\!\!&\mathbb{E}^{x_0}_\Phi\left(\mathbb{E}^{x_0}\left(\Delta^2|\Phi\right)\right)-\widetilde{D}^2(p) \nonumber \\
&\!\!\!\!\stackrel{(b)}{=}\!\!\!\!&\mathbb{E}^{x_0}_\Phi\bigg(\frac{2-\mathbb{P}^{x_0}(\mathcal{C}_\Phi)}{(\mathbb{P}^{x_0}(\mathcal{C}_\Phi))^2}\bigg)-\widetilde{D}^2(p) \nonumber\\
&\!\!\!\!=\!\!\!\!&2\mathbb{E}^{x_0}_\Phi\bigg(\frac{1}{(\mathbb{P}^{x_0}(\mathcal{C}_\Phi))^2}\bigg)-\widetilde{D}(p)-\widetilde{D}^2(p),\nonumber\\ &&\label{equ:vartmp_aloha}
\end{eqnarray}
where $(a)$ follows from the total expectation formula, and $(b)$ follows from the second moment of the geometrically distributed random variable.
From (\ref{equ:succ_aloha}), we have
\begin{eqnarray}
\!\!\!\!\!\!&&\!\!\!\!\!\!\mathbb{E}^{x_0}_\Phi\left(\frac{1}{(\mathbb{P}^{x_0}(\mathcal{C}_\Phi))^2}\right) \nonumber \\
\!\!\!\!\!\!&\!\!\!\!\!\!=\!\!\!\!\!\!\!&p^{-2}\exp\big({2\theta r_0^\alpha WN_0}\big)\nonumber\\
\!\!\!\!\!\!&&\mathbb{E}^{x_0}_\Phi\bigg(\frac{1}{\prod_{x\in\Phi\backslash\{x_0\}}\big(p\frac{1}{1+\theta r_0^\alpha|x|^{-\alpha}}+1-p\big)^2}\bigg) \nonumber \\
\!\!\!\!\!\!&\!\!\!\!\!\!\stackrel{(c)}{=}\!\!\!\!\!\!\!&p^{-2}\exp\bigg({2\theta r_0^\alpha WN_0}\nonumber\\
\!\!\!\!\!\!&&-\lambda\int_{\mathbb{R}^d}\bigg(1-\frac{1}{\big(p\frac{1}{1+\theta r_0^\alpha|x|^{-\alpha}}+1-p\big)^2}\bigg)\mathrm{d}x\bigg) \nonumber \\
\!\!\!\!\!\!&\!\!\!\!\!\!=\!\!\!\!\!\!\!&\frac{1}{p^2}\exp\left(2\theta r_0^\alpha WN_0+\frac{\lambda c_d r_0^d\theta^\delta C(\delta)(2-p-\delta p)p}{(1-p)^{2-\delta}}\right),\nonumber\\
\!\!\!\!\!\!&& \label{equ:quadprob_aloha}
\end{eqnarray}
where $(c)$ follows from the probability generating functional (PGFL) of the PPP.
Plugging (\ref{equ:quadprob_aloha}) into (\ref{equ:vartmp_aloha}) and applying the definition of $A$ and $B$ in Theorem \ref{thm:ld fh}, we get the variance of the local delay in Theorem \ref{thm:var aloha}.
%\begin{eqnarray}
%&=&\begin{cases}
%\infty&p=1\\
%2(\log_{1+\theta}2)^2p^{-2}\exp\Big(\lambda c_d r_0^d\theta^\delta C(\delta)(2p^{-1}-1-\delta)p^2(1-p)^{\delta-2}+ &\\
%\qquad\qquad\qquad\quad+2\theta r_0^\alpha WN_0\Big)-\widetilde{D}(p)\log_{1+\theta}2-\widetilde{D}^2(p)&0<p<1
%\end{cases}.
%\end{eqnarray}
\end{proof}

\subsection{Comparison Between FHMA and ALOHA}
\subsubsection{Mean local delay}
The mean local delay in FHMA is given by $D(N)$ in (\ref{equ:localdelay}), and that in ALOHA is given by $\widetilde{D}(p)$ in (\ref{equ:localdelay_aloha}).
In ALOHA, if the transmit probability $p$ is set as $\frac{1}{N}$, by comparing $\widetilde{D}(\frac{1}{N})$ to the result of FHMA, $D(N)$ given in (\ref{equ:localdelay}), we observe that the only difference lies in the thermal noise term.
In FHMA, the entire frequency band is divided into a number of sub-bands, thus reducing the noise power. However, with ALOHA, the noise scaling effect does not exist since the entire frequency band is used for transmission.
The mean local delays of the two schemes are the same if we ignore the thermal noise term and set $p=\frac{1}{N}$.
\subsubsection{Variance of the local delay}
Comparing (\ref{equ:var}) and (\ref{equ:var_aloha}), we observe that even if the noise is ignored and the transmit probability is set as $p=\frac{1}{N}$ in ALOHA, there is still an significant difference between the variances of the two schemes when $N>1$: in FHMA, a factor $N\left(N+1\right)$ exists in the first term, while in ALOHA the factor is $2N^2$.
When $N>1$, we have $V(N)<\widetilde{V}(\frac{1}{N})$ and this illustrates that the variance of the local delay for FHMA is less than that for ALOHA (see Fig. \ref{fig:variance}).
We further observe from Fig. \ref{fig:variance} that when $N\rightarrow\infty$, the variance for FHMA stabilizes at a typically small value, while for ALOHA, the variance increases quickly with $N$.
To understand the limiting characteristics quantitatively, we evaluate how the variances of the local delay scale with $N$ in the following proposition.
\begin{prop}
For FHMA with number of sub-bands $N>1$, the variance of the local delay is
\begin{eqnarray}
V(N)=(2-\delta)A+B+O\bigg(\frac{1}{N}\bigg)=\Theta(1). \label{equ:var_lim}
\end{eqnarray}
For ALOHA with transmit probability $p=\frac{1}{N}<1$, the variance of the local delay is
\begin{eqnarray}
\widetilde{V}(\frac{1}{N})=\Theta(N^2). \label{equ:var_lim_aloha}
\end{eqnarray}
\end{prop}
\begin{proof}
For FHMA with number of sub-bands $N>1$, from (\ref{equ:localdelay}), we have
\begin{eqnarray}
D(N)&\!\!\!\!=\!\!\!&N\exp\bigg(\frac{A+B}{N}+\frac{(1-\delta)A}{N^2}+O\bigg(\frac{1}{N^3}\bigg)\bigg) \nonumber \\
&\!\!\!\!=\!\!\!&N+(A+B) \nonumber\\
&\!\!\!\!\!\!\!&\!\!\!\!+\bigg(\frac{(A+B)^2}{2}+(1-\delta)A\bigg)\frac{1}{N}+O\bigg(\frac{1}{N^2}\bigg). \nonumber \\
&&\label{equ:var_lim1}
\end{eqnarray}
Then, $D^2(N)$ is given by
\begin{eqnarray}
D^2(N)&=&N^2+2(A+B)N+2(A+B)^2 \nonumber \\
&&+2(1-\delta)A+O\left(\frac{1}{N}\right). \label{equ:var_lim2}
\end{eqnarray}
From (\ref{equ:var}), we have the variance of the local delay as
\begin{eqnarray}
V(N)&\!\!\!\!=\!\!\!&N(N+1)\exp\bigg(\frac{2B}{N}+\frac{2A}{N}-\frac{3(\delta-1)A}{N^2} \nonumber\\
&\!\!\!\!\!\!\!\!&+O\bigg(\frac{1}{N^3}\bigg)\bigg)-D(N)-D^2(N) \nonumber \\
&\!\!\!\!=\!\!\!&N^2+(2A+2B+1)N+2(A+B)\nonumber\\
&\!\!\!\!\!\!\!\!&+2(A+B)^2-3(\delta-1)A\nonumber\\
&\!\!\!\!\!\!\!\!&-D(N)-D^2(N)+O\left(\frac{1}{N}\right). \label{equ:var_lim3}
\end{eqnarray}
Plugging (\ref{equ:var_lim1}) and (\ref{equ:var_lim2}) into (\ref{equ:var_lim3}), we get the result in (\ref{equ:var_lim}). The derivations of the limiting for ALOHA is similar, and we omit the details of that proof.
\end{proof}

\begin{figure}[!ht]
\centering
\includegraphics[width=0.45\textwidth]{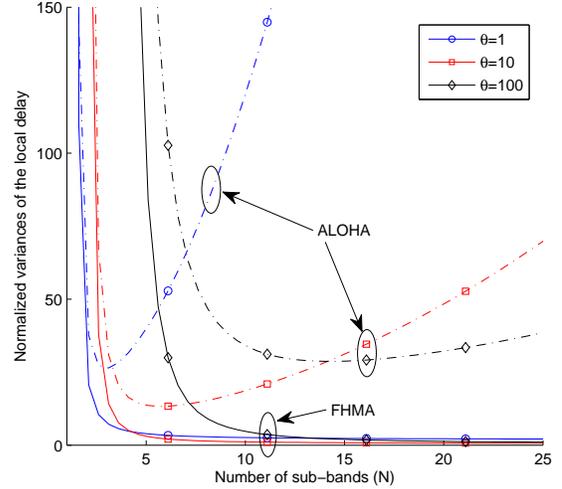}
\caption{Normalized variances of the local delay for FHMA, $\frac{V(N)}{(\log_2(1+\theta))^2}$, and for ALOHA, $\frac{\widetilde{V}(\frac{1}{N})}{(\log_2(1+\theta))^2}$, as a function of the number of sub-bands $N$, when $d=2$, $\lambda=0.01\mathrm{m}^{-2}$, $\alpha=4$, $r=5$m, and thermal noise ignored.}
\label{fig:variance}
\end{figure}

\subsection{Finiteness of the Mean Local Delay}
%\subsection{The case when $N=1$}
To understand why the mean local delay goes to infinity when $N$ is set to one, let us consider the expression for the mean local delay $D=N\mathbb{E}^{x_0}_\Phi\left(\frac{1}{\mathbb{P}^{x_0}(\mathcal{C}_\Phi)}\right)$ in FHMA, where $\frac{1}{\mathbb{P}^{x_0}(\mathcal{C}_\Phi)}$ is a random variable with support set $(1,+\infty)$ because it is the reciprocal of the successful transmit probability conditioned upon the PPP $\Phi$.
When $N=1$, the expectation $\mathbb{E}^{x_0}_\Phi\left(\frac{1}{\mathbb{P}^{x_0}(\mathcal{C}_\Phi)}\right)$ is infinity because the ccdf of $\frac{1}{\mathbb{P}^{x_0}(\mathcal{C}_\Phi)}$ has a heavy tail.
To show the heavy tail behavior, let us derive a lower bound for the ccdf of the local delay when $N=1$. Ignoring the thermal noise term in (\ref{equ:succ}), for $N=1$ and any $t\in(1,+\infty)$, we have
\begin{eqnarray}
\!\!\mathbb{P}^{x_0}\bigg(\frac{1}{\mathbb{P}^{x_0}(\mathcal{C}_\Phi)}>t\bigg)&\!\!\!\!=\!\!\!\!&\mathbb{P}^{x_0}\bigg(\!\!\!\prod_{x\in\Phi\backslash\{x_0\}}\!\!\!\!\!\big(1+\theta r_0^\alpha|x|^{-\alpha}\big)>t\bigg) \nonumber \\
\!\!&\!\!\!\!>\!\!\!\!&\mathbb{P}^{x_0}\left(1+\theta r_0^\alpha|x_{\mathrm{min}}|^{-\alpha}>t\right)\nonumber \\
\!\!&\!\!\!\!>\!\!\!\!&\mathbb{P}^{x_0}\left(\theta r_0^\alpha|x_{\mathrm{min}}|^{-\alpha}>t\right)\nonumber \\
\!\!&\!\!\!\!=\!\!\!\!&\mathbb{P}^{x_0}\left(|x_{\mathrm{min}}|<\theta^{\frac{1}{\alpha}}t^{-\frac{1}{\alpha}}r_0\right), \label{equ:1}
\end{eqnarray}
where $x_{\mathrm{min}}=\mathrm{argmin}_{x\in\Phi\setminus\{x_0\}}|x|$ is the nearest interfering transmitter to the receiver.
The distance between the receiver and its nearest interfering transmitter has cumulative distribution function (cdf) as $b(r)=1-\exp(- c_d\lambda r^d)$. Substituting this cdf into (\ref{equ:1}) and letting $\delta=\frac{2}{\alpha}$, we get
\begin{eqnarray}
\!\!\!\!\mathbb{P}^{x_0}\left(\frac{1}{\mathbb{P}^{x_0}(\mathcal{C}_\Phi)}>t\right)&\!\!\!\!>\!\!\!\!&1-\exp\left(-c_d\lambda\theta^\delta t^{-\delta}r_0^d\right) \qquad \label{equ:ccdf_lowerb} \\
\!\!\!\!&\!\!\!\!\sim\!\!\!\!&\frac{c_d\lambda\theta^\delta r_0^d}{t^{\delta}}, \quad t\rightarrow\infty.
\end{eqnarray}
The function $g(t)=1-\exp\left(-C_0 t^{-\delta}\right)$, where $C_0=c_d\lambda\theta^\delta r_0^d$, gives a lower bound for the ccdf of the local delay when $N=1$ (see Fig. \ref{fig:lbound}).

By the identity of $\mathbb{E}(X)=\int_0^\infty\mathbb{P}(X>t)\mathrm{d}t$ for any non-negative random variable $X$ and the inequality $e^{-x}<1-x+\frac{x^2}{2}$ for $x>0$, we have
\begin{eqnarray}
\!\!\!\!\!\!\!\!\mathbb{E}^{x_0}_\Phi\left(\frac{1}{\mathbb{P}^{x_0}(\mathcal{C}_\Phi)}\right)&\!\!\!\!=\!\!\!\!&\int_0^\infty\left(\mathbb{P}^{x_0}\left(\frac{1}{\mathbb{P}^{x_0}(\mathcal{C}_\Phi)}>t\right)\right)\mathrm{d}t \nonumber \\
\!\!\!\!\!\!\!\!&\!\!\!\!=\!\!\!\!&1+\int_1^\infty\left(\mathbb{P}^{x_0}\left(\frac{1}{\mathbb{P}^{x_0}(\mathcal{C}_\Phi)}>t\right)\right)\mathrm{d}t \nonumber \\
\!\!\!\!\!\!\!\!&\!\!\!\!>\!\!\!\!&1+\int_1^\infty \left(1-\exp\left(-C_0 t^{-\delta}\right)\right)\mathrm{d}t \nonumber\\
\!\!\!\!\!\!\!\!&\!\!\!\!>\!\!\!\!&1+\int_1^\infty\left( C_0t^{-\delta}-\frac{1}{2}C_0^2t^{-2\delta}\right)\mathrm{d}t\nonumber\\
\!\!\!\!\!\!\!\!&\!\!\!\!=\!\!\!\!&\infty.
\end{eqnarray}

%The penultimate inequality follows because $e^{-x}<1-x+\frac{x^2}{2}$ for $0<x<1$.
%Since $e^{-x}>1-x+\frac{x^2}{2}$ for $x>0$, we obtain
%\begin{eqnarray}
%\mathbb{P}^{x_0}\left(\frac{1}{\mathbb{P}(\mathcal{C}_\Phi)}>t\right)&>&c_0t^{-\delta}-\frac{1}{2}c_0^2t^{-2\delta}.
%\end{eqnarray}

\begin{figure}[!ht]
\centering
\includegraphics[width=0.5\textwidth]{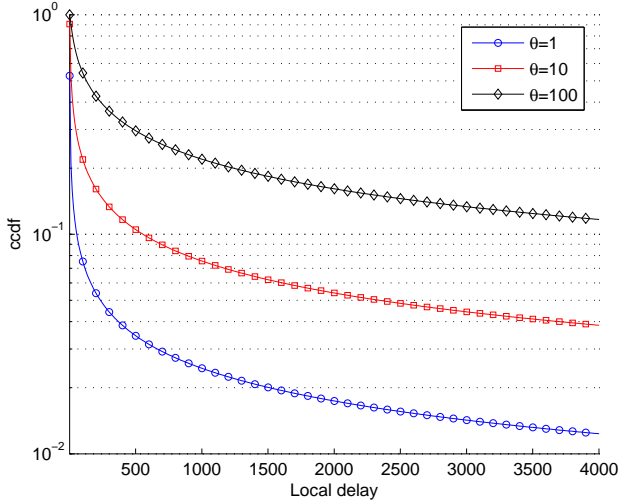}
\caption{Lower bound for the ccdf of the local delay, given by (\ref{equ:ccdf_lowerb}), when $N=1$ in the 2-dimensional case ($d=2$). The intensity of transmitters is $\lambda=0.01\mathrm{m}^{-2}$ and the path loss exponent is $\alpha=4$.}
\label{fig:lbound}
\end{figure}

When FHMA is applied with $N>1$, there is an additional term $\frac{N-1}{N}$ in the success probability given by (\ref{equ:succ}), which prevents $\mathbb{P}^{x_0}(\mathcal{C}_\Phi)$ from getting too small when $|x|$ approaches zero.
It can be interpreted intuitively that, although there are some interfering transmitters very close to the receiver, the application of FHMA guarantees that there is always a relatively large probability that those transmitters do not continuously
cause interference to the receiver.

\subsection{The effect of bounded path loss function}
\label{sec:bounded}
In the discussion above, we have considered the unbounded path loss function $l(r)=\kappa r^{-\alpha}$.
Though the unbounded path loss function is an idealized model, it gives an effective approximation to the actual path loss and results in concise results \cite{inaltekin2009unbounded,gulati2012joint}.
In this subsection, we compare the results derived under the unbounded path loss function to that under the bounded path loss function $l(r)=\kappa(r^{\alpha}+\varepsilon)^{-1}$, where $\varepsilon>0$.
The unbounded path loss function is the limiting case of the bounded path loss function as $\varepsilon\rightarrow0$.

Without loss of generality, we take FHMA as an example. By replacing $l(r)=\kappa r^{-\alpha}$ with $l(r)=\kappa(r^{\alpha}+\varepsilon)^{-1}$ in the derivations of Theorem \ref{thm:ld fh} and Theorem \ref{thm:var fp}, we obtain the mean and variance of the local delay under the assumption of bounded path loss function as follows.

\begin{eqnarray}
D_\varepsilon(N)=N\exp\bigg(\frac{(r_0^\alpha+\varepsilon)\theta WN_0}{N}\nonumber\\
+\frac{\lambda c_d (r_0^\alpha+\varepsilon)\theta C(\delta)}{(N\varepsilon+ (N-1)\theta(r_0^\alpha+\varepsilon))^{1-\delta}N^{\delta}}\bigg), \label{eqn:D_epsilon}
\end{eqnarray}

\begin{eqnarray}
V_\varepsilon(N)=N\left(N+1\right)
\exp\bigg(\frac{2\theta r_0^\alpha WN_0}{N}\nonumber\\
+\frac{(2N\varepsilon+(2N-1-\delta)\theta(r_0^\alpha+\varepsilon))\lambda c_d (r_0^\alpha+\varepsilon)\theta C(\delta)}{N^{\delta}(N\varepsilon+(N-1)\theta(r_0^\alpha+\varepsilon))^{2-\delta}}\bigg)\nonumber
\end{eqnarray}
\begin{equation}
\qquad\qquad-D_\varepsilon(N)-D_\varepsilon^2(N). \label{eqn:V_epsilon}
\end{equation}

%The equation (\ref{eqn:D_epsilon}) and (\ref{eqn:V_epsilon}) are applicable to all $N$ including $N=1$.
%When $N>1$, it can be verified that $\lim_{\varepsilon\rightarrow0}D_\varepsilon(N)=D(N)$ and $\lim_{\varepsilon\rightarrow0}V_\varepsilon(N)=V(N)$.
%When $N=1$, both $D_\varepsilon(1)$ and $V_\varepsilon(1)$ are finite; however, as $\varepsilon\rightarrow0$ their limits diverge, thus returning to the unbounded path loss case.

When $N=1$, both $D_\varepsilon(1)$ and $V_\varepsilon(1)$ are finite if $\varepsilon>0$.
Setting $\varepsilon=0$ reproduces the results for the unbounded model, as expected.
As can be seen, the difference between the results for the unbounded model and the bounded one decreases with increasing $r_0$ or decreasing $\theta$.
In order to evaluate the difference, we set $\varepsilon$ as the typical value $\kappa$ (i.e., the path loss becomes $l(r)=\kappa (r^{\alpha}+\kappa)^{-1}$) such that the received power never exceeds the transmitted one without fading.
The value of $\kappa$ is the path loss at $1$m TX-RX separation, which is rather small, typically like $-30$dB \cite[Ch. 3]{rappaport1996wireless}.
Therefore, as $\varepsilon=\kappa\rightarrow0$, the mean local delay for $N=1$ is approximated as
\begin{equation}
D_\varepsilon(1)\sim\exp\left(\theta r_0^\alpha WN_0 + \frac{\lambda c_d r_0^\alpha \theta C(\delta)}{\varepsilon^{1-\delta}}\right), \quad \varepsilon\rightarrow 0.
\end{equation}
It is observed that $D_\varepsilon(1)$ increases exponentially with respect to $1/\epsilon^{1-\delta}$ as $\varepsilon\rightarrow0$.
For the realistic bounded path loss model, in which $\varepsilon$ is rather small, the mean local delay when $N=1$ is finite though extremely large.
Thus, we can conclude that for the realistic bounded path loss model, when $N>2$ the boundedness of the path loss has only negligible effect on the mean and the variance of the local delay; when $N=1$ the mean local delay is extremely large and thus can be considered as infinity for practical purposes.

\section{Optimal Parameters To Minimize Mean Local Delay}
\label{sec:opt_par}
In this section, we analyze the optimal number of sub-bands in FHMA and optimal transmit probability in ALOHA to minimize the mean local delay.
%in the 2-dimensional case, which means that $d=2$, $c_d=\pi$ and $0<\delta<1$.
Deriving the optimal parameters is difficult, and the results may not be compact; thus, we resort to deriving tight bounds for the optimal values.
\subsection{FHMA}
In the derivation, we relax $N$ to be continuous and subsequently take the actual optimal number to be a nearby integer.
The following theorem gives the bounds of the optimal number of sub-bands.
\begin{thm}
The bounds of the optimal number of sub-bands that minimizes the mean local delay are given by
%\begin{eqnarray}
%\!\!\!\!\!\!\!\!&&\!\!\!\!N_{\mathrm{opt}}\in(t_0,t_0+2), t_0=\lambda c_d r_0^d\theta^{\delta}C(\delta)+\theta r_0^\alpha WN_0.\nonumber\\
%\!\!\!\!\!\!\!\!&&\!\!\!\!\label{equ:thmbounds}
%\end{eqnarray}
%The actual optimal number is an integer, and the bounds are as follows:
\begin{eqnarray}
\!\!\!\!\!\!\!\!&&\!\!\!\!N_{\mathrm{opt}}\in[\lfloor t_0\rfloor,\lceil t_0\rceil+2],\quad t_0=\lambda c_d r_0^d\theta^{\delta}C(\delta)+\theta r_0^\alpha WN_0.\nonumber\\
\!\!\!\!\!\!\!\!&&\!\!\!\! \label{equ:thmbounds}
\end{eqnarray}
\end{thm}
\begin{proof}
Based on the result of (\ref{equ:localdelay}), we get the derivative of the mean local delay $D'(N)$ when $N>1$ as follows
\begin{equation}
D'(N)=f(N)\exp\big(\lambda c_dr_0^d\theta^{\delta}C(\delta)(N-1)^{\delta-1}N^{-\delta}\nonumber
\end{equation}
\begin{equation}
 \qquad +\theta r_0^\alpha WN_0N^{-1}\big),
\end{equation}
where
\begin{equation}
f(N)=1-\frac{1}{N}\bigg(\lambda c_d r_0^d\theta^{\delta}C(\delta)\bigg(\frac{N}{N-1}\bigg)^{1-\delta}\frac{N-\delta}{N-1}\nonumber
\end{equation}
\begin{equation}
+\theta r_0^\alpha WN_0\bigg).
\end{equation}
We observe that $f(N)$ is strictly monotonically increasing in $N$; this means that there is only one optimal value $N_{\mathrm{opt}}$ that satisfies $D'(N_{\mathrm{opt}})=0$, which is given by $f(N_{\mathrm{opt}})=0$. From $f(N_{\mathrm{opt}})=0$, we get
\begin{eqnarray}
N_{\mathrm{opt}}&=&\lambda c_d r_0^d\theta^{\delta}C(\delta)\left(\frac{N_{\mathrm{opt}}}{N_{\mathrm{opt}}-1}\right)^{2-\delta}\frac{N_{\mathrm{opt}}-\delta}{N_{\mathrm{opt}}}\nonumber\\
&&+\theta r_0^\alpha WN_0. \label{equ:Nopt}
\end{eqnarray}
Since $N_{\mathrm{opt}}/(N_{\mathrm{opt}}-1)>1$ and $0<\delta<1$, we have
\begin{eqnarray}
N_{\mathrm{opt}}&>&\lambda c_d r_0^d\theta^{\delta}C(\delta)\left(\frac{N_{\mathrm{opt}}}{N_{\mathrm{opt}}-1}\right)^{2-\delta}\frac{N_{\mathrm{opt}}-1}{N_{\mathrm{opt}}}\nonumber\\
&&+\theta r_0^\alpha WN_0 \nonumber \\
&>&\lambda c_d r_0^d\theta^{\delta}C(\delta)+\theta r_0^\alpha WN_0, \label{equ:lower}
\end{eqnarray}
This gives a lower bound for the optimal value $N_{\mathrm{opt}}$. Next we derive an upper bound for $N_{\mathrm{opt}}$.
From (\ref{equ:Nopt}), we have
\begin{eqnarray}
N_{\mathrm{opt}}&<&\lambda c_d r_0^d\theta^{\delta}C(\delta)\left(\frac{N_{\mathrm{opt}}}{N_{\mathrm{opt}}-1}\right)^2+\theta r_0^\alpha WN_0 \nonumber \\
&<&\left(\lambda c_d r_0^d\theta^{\delta}C(\delta)+\theta r_0^\alpha WN_0\right)\left(\frac{N_{\mathrm{opt}}}{N_{\mathrm{opt}}-1}\right)^2.  \nonumber
\end{eqnarray}
Then, we have
\begin{eqnarray}
\frac{(N_{\mathrm{opt}}-1)^2}{N_{\mathrm{opt}}}&<&\lambda c_d r_0^d\theta^{\delta}C(\delta)+\theta r_0^\alpha WN_0. \nonumber
\end{eqnarray}
\begin{eqnarray}
N_{\mathrm{opt}}-2+\frac{1}{N_{\mathrm{opt}}}&<&\lambda c_d r_0^d\theta^{\delta}C(\delta)+\theta r_0^\alpha WN_0. \nonumber
\end{eqnarray}
\begin{eqnarray}
N_{\mathrm{opt}}<&\lambda c_d r_0^d\theta^{\delta}C(\delta)+\theta r_0^\alpha WN_0+2. \label{equ:upper}
\end{eqnarray}
Combining (\ref{equ:lower}) and (\ref{equ:upper}) and noting that $N$ is an integer, we get the bounds of $N_{\mathrm{opt}}$.
% as follows
%\begin{eqnarray}
%N_{\mathrm{opt}}\in(t_0,t_0+2),\quad t_0=\lambda c_d r_0^d\theta^{\delta}C(\delta)+\theta r_0^\alpha WN_0 \label{equ:bounds}
%\end{eqnarray}
\end{proof}

The bounds given here are rather simple and tight.
%However, since the number $N_{\mathrm{opt}}$ should be an integer, the actual bounds should be rounded and given as $[\lfloor t_0\rfloor,\lceil t_0\rceil+2]$.
If frequency splitting is not applied (the case of $N=1$), the mean local delay will surely be infinite. To guarantee a finite mean local delay, the value of $N$ should be at least two. It is valuable to investigate for which range of parameters the optimal value $N_{\mathrm{opt}}$ will be two. The following corollary gives such a condition.
\begin{cor}
If the intensity of transmitters $\lambda$ satisfies the following inequality,
\begin{eqnarray}
\lambda&<&\frac{\ln\frac{3}{2}-\frac{1}{6}\theta r_0^\alpha WN_0}{c_dr_0^d\theta^\delta C(\delta)(2^{-\delta}-2^{\delta-1}3^{-\delta})},\label{equ:condition}
\end{eqnarray}
the optimal number of sub-bands $N_{\mathrm{opt}}$ that minimizes the mean local delay will be two.
\end{cor}
\begin{proof}
Since we have proved that mean local delay $D(N)$ is a function that first decreases and then increases with $N$, the condition for $N_{\mathrm{opt}}=2$ is $D(2)<D(3)$. By substituting the expression of mean local delay (\ref{equ:localdelay}) into $D(2)<D(3)$, we get the condition in the corollary.
Notice that the right side of the inequality in (\ref{equ:condition}) may be negative, in which case the condition for $\lambda$ cannot be satisfied and then the optimal number of sub-bands will never be two.
\end{proof}

In Fig. \ref{fig:Nopt_bound}, we plot the optimal number of sub-bands $N_{\mathrm{opt}}$ and its bounds given in (\ref{equ:thmbounds}) as a function of the path loss exponent for different $\theta$. The optimal number $N_{\mathrm{opt}}$ is obtained by numerical calculation of the solution of the equation (\ref{equ:Nopt}). We observe from Fig. \ref{fig:Nopt_bound} that the bounds are quite tight and give excellent approximation of the value $N_{\mathrm{opt}}$. This figure also shows that the optimal value $N_{\mathrm{opt}}$ decreases with increasing path loss exponent, which verifies our aforementioned discussion regarding Fig. \ref{fig:change}. The curves show that when the path loss exponent is fixed, the optimal number $N_{\mathrm{opt}}$ is an increasing function of the SINR threshold $\theta$, which can be also perceived from the expression (\ref{equ:thmbounds}). This is reasonable since with larger SINR threshold, the condition for successful transmission becomes harsher, and more sub-bands are needed to meet the stronger requirements. In Fig. \ref{fig:D_Nopt}, we plot the minimum value of the normalized mean local delay when the optimum number of sub-bands $N_{\mathrm{opt}}$ is used. We observe from Fig. \ref{fig:D_Nopt} that there are intersection points between different curves, implying that the choice of SINR threshold $\theta$ has direct impact on the mean local delay.
In Section \ref{sec:opt_thd}, we will try to obtain the optimal SINR threshold.

\begin{figure}[!ht]
\centering
\includegraphics[width=0.5\textwidth]{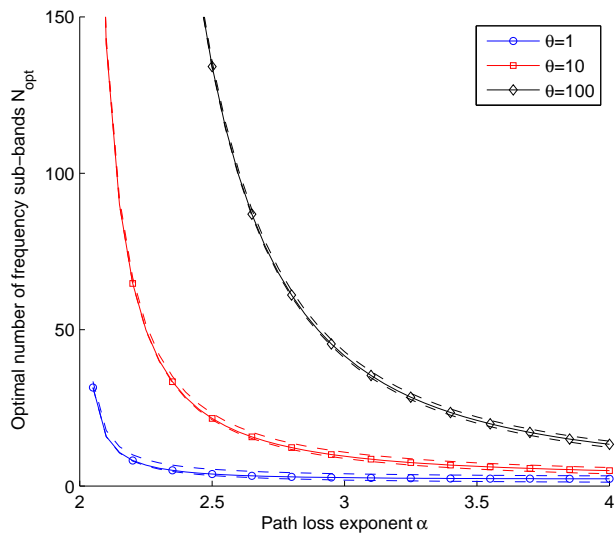}
\caption{Optimal number of sub-bands $N_{\mathrm{opt}}$ and its bounds $(t_0,t_0+2)$ as a function of the path loss exponent for varying $\alpha$.}
\label{fig:Nopt_bound}
\end{figure}

\begin{figure}[!ht]
\centering
\includegraphics[width=0.5\textwidth]{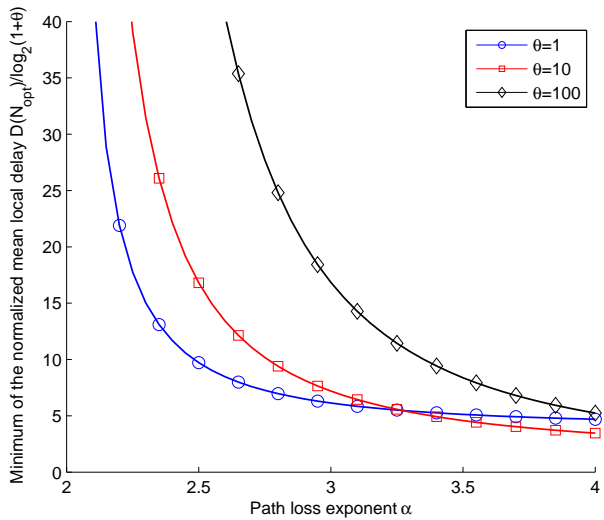}
\caption{Minimum of the normalized mean local delay $\frac{D(N_{\mathrm{opt}})}{\log_2(1+\theta)}$ as a function of the path loss exponent $\alpha$.}
\label{fig:D_Nopt}
\end{figure}
\subsection{ALOHA}
Since the expressions of the mean local delays for ALOHA and for FHMA are the same when $p=1/N$ and thermal noise ignored, we give the optimal transmit probability in the following theorem directly and omit the proof.
\begin{thm}
The bounds of the optimal transmit probability which minimizes the mean local delay is
\begin{eqnarray}
p_{\mathrm{opt}}\in\left(\frac{1}{\lambda c_d r_0^d\theta^{\delta}C(\delta)+2},\frac{1}{\lambda c_d r_0^d\theta^{\delta}C(\delta)}\right). \label{equ:thmbounds_aloha}
\end{eqnarray}
\end{thm}

\section{Optimal SINR Threshold $\theta$}
\label{sec:opt_thd}
%The minimization of the mean local delay is a joint optimization problem with respect to parameters $N$ (or $p$ in ALOHA case) and $\theta$.
In the discussion above, we have already derived tight bounds for the optimal number of sub-bands $N_{\mathrm{opt}}$
and the optimal transmit probability $p_{\mathrm{opt}}$ to minimize the mean local delay when the SINR threshold $\theta$ is fixed.
%From (\ref{equ:thmbounds}), if we approximate the optimal number $N_{\mathrm{opt}}$ as its lower bound, we obtain that $N_{\mathrm{opt}}\sim\lambda c_d r_0^d\theta^{\delta}C(\delta)+\theta r_0^\alpha WN_0$, which is an increasing function of the threshold $\theta$.
%The relationship between the optimal number $N_{\mathrm{opt}}$ and $\theta$ is shown in Fig. \ref{fig:Nopttheta}.
%The minimum value of local delay under changing SINR threshold $\theta$ is shown in Fig. \ref{fig:DNopttheta}.
%The results show that there is a optimal threshold $\theta$ that minimizes local delay.
The following analysis will focus on deriving the optimal threshold $\theta_{\mathrm{opt}}$ or its bounds when the number of sub-bands $N$ or the transmit probability $p$ is fixed.
However, as mentioned in Section \ref{sec:model}, the duration of each time slot is proportional to $\frac{1}{\log_2(1+\theta)}$.
In order to characterize the actual delay, we slightly modify the optimization objective as the normalized mean local delay, i.e., $\frac{D(N)}{\log_2(1+\theta)}$ for FHMA and $\frac{\widetilde{D}(p)}{\log_2(1+\theta)}$ for ALOHA.
%However, the optimization of $\theta$ is different from the optimization of $N$ or $p$ since the length of each time slot is proportional to $\frac{1}{\log_2(1+\theta)}$. The local delay, measured by the number of time slots, is not enough to characterize the actual delay. Therefore, in the following discussion, we slightly modify the optimization objective as the mean local delay normalized by $\log_2(1+\theta)$: $\frac{D(N)}{\log_2(1+\theta)}$ for FHMA and $\frac{\widetilde{D}(p)}{\log_2(1+\theta)}$ for ALOHA.
In the following analysis, we consider two asymptotic regimes: the interference-limited regime and the noise-limited regime.
The interference-limited regime is typically encountered in cellular radio systems like CDMA networks, where the interference dominates over the thermal noise.
The noise-limited regime is appropriate if the distance between concurrent transmitters is much larger than the distance of the typical link, in which case the interference in the network is negligible.

\subsection{FHMA}
The following theorem gives the optimal threshold $\theta_{\mathrm{opt}}$ and its bounds for FHMA in the interference-limited regime and noise-limited regime respectively.
\begin{thm}
\label{thm:theta_fh}
In the noise-limited regime, the optimal threshold $\theta_{\mathrm{opt}}$ that minimizes the normalized mean local delay $\frac{D(N)}{\log_2(1+\theta)}$ for FHMA is given by
\begin{eqnarray}
\theta_\mathrm{opt}&=&\exp\left(\mathcal{W}\left(\frac{N}{r_0^{\alpha} WN_0}\right)\right)-1, \label{equ:threshold}
\end{eqnarray}
where $\mathcal{W}(z)$ is the Lambert $\mathcal{W}$ function which solves $\mathcal{W}(z)e^{\mathcal{W}(z)}=z$.
In the interference-limited regime, the bounds of $\theta_{\mathrm{opt}}$ are given by
\begin{eqnarray}
\theta_{\mathrm{opt}}\in\left(b_0^{-1/(\delta+1)}-1,b_0^{-1/\delta}\right), \nonumber \\
b_0=\lambda c_d r_0^d\delta C(\delta)(N-1)^{\delta-1}N^{-\delta}.\label{equ:threshold_bounds}
\end{eqnarray}
\end{thm}
\begin{proof}
The derivative of the normalized mean local delay with respective to $\theta$ is
\begin{eqnarray}
\frac{\partial \left(\frac{D(N)}{\log_2(1+\theta)}\right)}{\partial\theta}&\!\!\!=\!\!\!&\frac{h(\theta)\theta^{\delta-1}N}{\log_2(1+\theta)}\exp\big(\theta r_0^\alpha WN_0N^{-1}\nonumber\\
&&\!\!\!\!\!\!\!\!\!\!\!\!\!\!\!\!+\lambda c_d r_0^d\theta^{\delta}C(\delta)(N-1)^{\delta-1}N^{-\delta}\big),
\end{eqnarray}
where $h(\theta)$ is as follows
\begin{eqnarray}
h(\theta)=\lambda c_d r_0^d\delta C(\delta)(N-1)^{\delta-1}N^{-\delta}\nonumber\\
+ r_0^\alpha WN_0N^{-1}\theta^{1-\delta}-\frac{1}{\theta^{\delta-1}(1+\theta)\ln(1+\theta)}.
\end{eqnarray}
Next, we prove that $h(\theta)$ is a strictly increasing function of $\theta$, then we show that the equation $h(\theta)=0$ has a unique solution. Let $l(\theta)=\theta^{\delta-1}(1+\theta)\ln(1+\theta)$ and the derivative of $l(\theta)$ is as follows
\begin{eqnarray}
l'(\theta)&=&((\delta-1)\theta^{\delta-2}+\delta\theta^{\delta-1})\ln(1+\theta)+\theta^{\delta-1}\nonumber \\
&>&((\delta-1)\theta^{\delta-2}+\delta\theta^{\delta-1})\theta+\theta^{\delta-1}\nonumber \\
%&=&\delta\theta^{\delta-1}+\delta\theta^{\delta}\nonumber \\
&>&0.\nonumber
\end{eqnarray}
Thus, $l(\theta)$ is a strictly increasing function of $\theta$. This implies that $h(\theta)$ is also a strictly increasing function of $\theta$. Since $\lim_{\theta\rightarrow0^+}h(\theta)=-\infty$ and $\lim_{\theta\rightarrow\infty}h(\theta)=+\infty$, the equation $h(\theta)=0$ has a unique solution $\theta_{\mathrm{opt}}$ that minimizes the mean local delay.

In the noise-limited regime, the equation $h(\theta)=0$ has the form
\begin{eqnarray}
r_0^\alpha WN_0N^{-1}\theta^{1-\delta}-\frac{1}{\theta^{\delta-1}(1+\theta)\ln(1+\theta)}=0.\quad
\end{eqnarray}
Solving this equation we obtain (\ref{equ:threshold}).

In the interference-limited regime, the noise is ignored, and $h(\theta)=0$ has the form
\begin{eqnarray}
\lambda c_d r_0^d\delta C(\delta)(N-1)^{\delta-1}N^{-\delta}\nonumber\\
-\frac{1}{\theta^{\delta-1}(1+\theta)\ln(1+\theta)}=0. \label{equ:h}
\end{eqnarray}
A closed-form solution for the above equation does not exist.
By applying the inequalities $\frac{\theta}{1+\theta}<\ln(1+\theta)<\theta$ to (\ref{equ:h}), we get the following inequalities
\begin{eqnarray}
\frac{1}{(1+\theta)^{\delta+1}}<\frac{1}{\theta^{\delta}(1+\theta)}\nonumber\\
<\lambda c_d r_0^d\delta C(\delta)(N-1)^{\delta-1}N^{-\delta}<\frac{1}{\theta^{\delta}}.
\end{eqnarray}
From the above inequalities, we get the bounds in (\ref{equ:threshold_bounds})
\end{proof}
%(The bounds in Theorem \ref{thm:theta_fh} seem very loose, I will try to figure out more tight bounds.)

\subsection{ALOHA}
Based on the similarity between the expressions of the normalized mean local delays for ALOHA and for FHMA, we obtain the optimal threshold for ALOHA directly.
\begin{thm}
In noise-limited regime, the optimal threshold $\theta_{\mathrm{opt}}$ that minimizes the normalized mean local delay for ALOHA  is as follows
\begin{eqnarray}
\theta_\mathrm{opt}&=&\mathcal{W}\left(\exp\left(\frac{1}{r_0^{\alpha} WN_0}\right)\right)-1. \label{equ:threshold_aloha}
\end{eqnarray}
In interference-limited regime, the bounds of $\theta_{\mathrm{opt}}$ are given by
\begin{eqnarray}
\theta_{\mathrm{opt}}\in\left(b_0^{-1/(\delta+1)}-1,b_0^{-1/\delta}\right),\nonumber\\
 b_0=\lambda c_d r_0^d\delta C(\delta)p(1-p)^{\delta-1}. \label{equ:threshold_bounds_aloha}
\end{eqnarray}
\end{thm}

\section{Design Insights}
\label{sec:design}
\subsection{Mean Delay-Jitter Tradeoff}
The jitter of delay, typically characterized by the packet delay variation, is defined in \cite{demichelis2002ip} and  \cite{demichelis2000packet}.
In the system design, the delay variation is an important measure that characterizes the fluctuation of delay \cite{152873}.
For interactive real-time applications, e.g., VoIP, large delay variance can be a serious issue.
To the best of our knowledge, the variance of local delay has not been explored in the existing work.
The optimal value of $N$ that minimizes the mean local delay is often not the one that minimizes the variance; thus there is a tradeoff between the mean and the variance of the local delay.
Fig. \ref{fig:delay_jitter} and Fig. \ref{fig:delay_jitter_aloha} visualize the relationship between the mean and the variance of the normalized local delay for FHMA and for ALOHA respectively.
From Fig. \ref{fig:delay_jitter} we observe that in FHMA the favorable operating point has a reasonably wide tuning range because the variance stabilizes fast as $N$ increases.
In contrast, we observe from Fig. \ref{fig:delay_jitter_aloha} that in ALOHA the curves turn sharply.

\begin{figure}[!ht]
  \centering
  \subfigure[FHMA]{
    \label{fig:delay_jitter} %% label for first subfigure
    \includegraphics[width=0.5\textwidth]{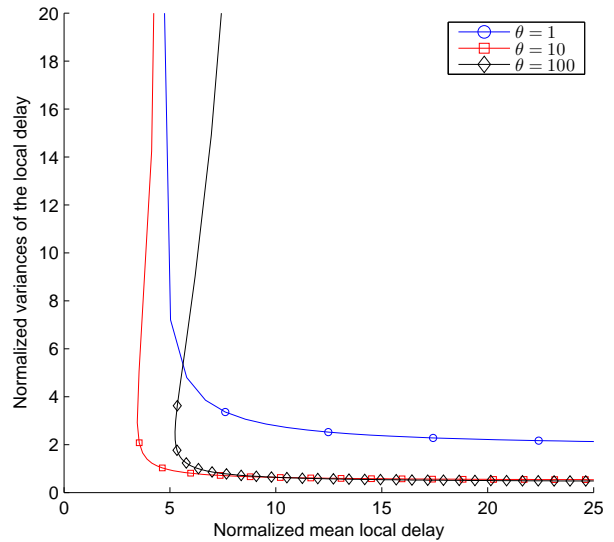}}
%  \hspace{0.2in}
  \subfigure[ALOHA]{
    \label{fig:delay_jitter_aloha} %% label for second subfigure
    \includegraphics[width=0.5\textwidth]{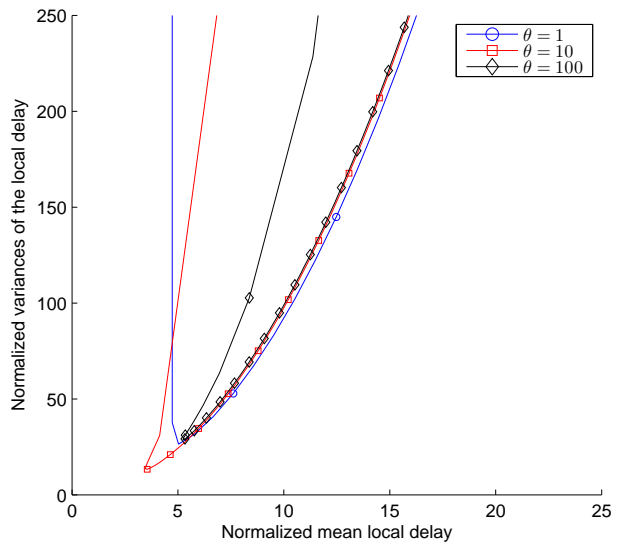}}
  \caption{Mean delay-jitter tradeoff.}
  \label{fig:delay_jitter_tradeoff} %% label for entire figure
\end{figure}

\subsection{Tail Probability of the Local Delay}
The tail probability is an important measure of the system performance since one may require (as a QoS constraint) that the probability that the local delay exceeds a certain threshold is less than a predefined value.
Based on the mean and variance we have derived and by applying the one-tailed Chebyshev's inequality, we obtain an upper bound for the tail probability.

For example, in FHMA with $N>1$, letting $X=\sum_{i=1}^{N}\Delta_i$ be the local delay, the tail probability is upper bounded as follows
\begin{equation}
\!\!\!\!\mathbb{P}\{X>T_0\}\leq\frac{V(N)}{V(N)+(T_0-D(N))^2}, \quad \mathrm{for} \quad T_0>D(N).\\
  \label{equ:chebyshev}
\end{equation}
For example, if we let the design requirement be that the probability that the local delay exceeds $10$ is less than $5\%$.
Then, when the threshold of the local delay is fixed as $T_0=10$, the upper bound of the tail probability given by (\ref{equ:chebyshev}) with varying $N$ is shown in Fig. \ref{fig:tailprob}.
We observe that in order to achieve the probability $5\%$, the number of sub-bands $N$ for the case when $\theta=1,10,100$ should be chosen larger than $15,50,95$ respectively.

The bounds based on Chebyshev's inequality will typically not provide the tightest bounds. However, it generally cannot be improved if only the mean and variance are available. On the other hand, if further statistical information is provided, a number of methods may be developed to improve the sharpness of the bounds, for example, through the use of semivariances if some samples are available, or through the use of Bhattacharyya's inequality or large-deviations based inequalities if higher moments or even the moment generating functions are available.

\begin{figure}[!ht]
\centering
\includegraphics[width=0.5\textwidth]{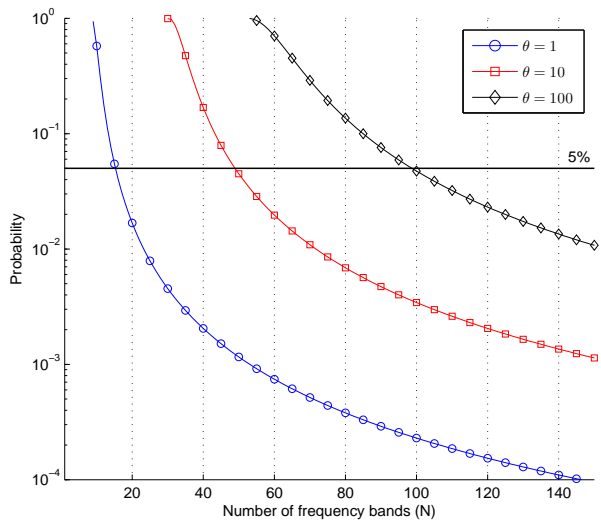}
\caption{Upper bounds of the tail probability of local delay given by (\ref{equ:chebyshev}) as a function of the number of sub-bands $N$ when fixing $T_0=10$ in FHMA.}
\label{fig:tailprob}
\end{figure}

\section{Conclusions}
\label{sec:conclusion}
In this work, we studied the problem of reducing the effect of interference correlation by introducing MAC dynamics.
We derived the mean and variance of the local delay and evaluated how the interference correlation can be reduced by FHMA and ALOHA.
We also evaluated the optimal number of sub-bands in FHMA and the optimal transmit probability in ALOHA that minimize the mean local delay.

The results reveal that there exist two operation regimes for the network, the correlation-limited regime and bandwidth-limited regime, which are separated by the optimal number of sub-bands in FHMA and the optimal transmit probability in ALOHA.
If no MAC dynamics is employed, the local delay has a heavy tail distribution which results in infinite mean local delay; meanwhile, employing FHMA and ALOHA will greatly decrease the mean local delay.
By comparing the results of FHMA and ALOHA, we observed that while the mean local delays of the two protocols are the same for certain parameters, the variances are rather different.
According to the results established herein, FHMA outperforms ALOHA if implementation costs like overhead are not taken into consideration; however, when considering the implementation costs, the overhead for FHMA may be much higher than ALOHA because in FHMA each transmitter should inform the corresponding receiver which sub-band to listen on.

\section*{Acknowledgement}

The authors wish to thank the anonymous reviewers for their constructive comments.

%%%%%%%%%%%%%%%%%%%%%%%%%%%%%%%%%%%%%%%%%%%%%%%%%%%%%%%%%%%%%%%%%%%%%%%%%%%%%%%%%%%%%%%%%%%%%%%%%%%%%%%%%%%%
\bibliographystyle{IEEEtran}
\bibliography{123}

\begin{IEEEbiography}[{\includegraphics[width=1in,height=1.25in,clip,keepaspectratio]{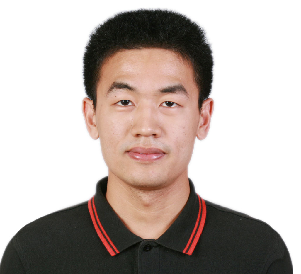}}]{Yi Zhong}
Yi Zhong received his B.S. degree in Electronic Engineering from University of Science and Technology of China (USTC) in 2010. He is now a Ph.D. student in Electronic Engineering at USTC, Hefei, China. From August to December 2012, he was a visiting student in Prof. Martin Haenggi's group at University of Notre Dame. From July to October 2013, he worked as an intern in Qualcomm, Corporate Research and Development, Beijing. His research interests include heterogeneous and femtocell-overlaid cellular networks, wireless ad hoc networks, stochastic geometry and point process theory.
\end{IEEEbiography}

\begin{IEEEbiography}[{\includegraphics[width=1in,height=1.25in,clip,keepaspectratio]{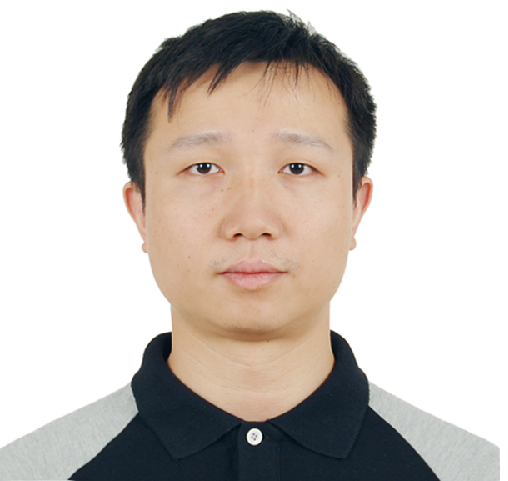}}]{Wenyi Zhang}
Wenyi Zhang (S'00, M'07, SM¡¯11) received the B.E. degree in automation from Tsinghua University, Beijing, China, in 2001, and the M.S. and Ph.D. degrees in electrical engineering both from the University of Notre Dame, Notre Dame, IN, in 2003 and 2006, respectively. He was affiliated with University of Southern California as a Postdoctoral Research Associate, and with the Qualcomm Corporate Research and Development, Qualcomm Incorporated. He is currently on the faculty of Department of Electronic Engineering and Information Science, University of Science and Technology of China. His research interests include wireless communications and networking, information theory, and statistical signal processing.
\end{IEEEbiography}

\begin{IEEEbiography}[{\includegraphics[width=1in,height=1.25in,clip,keepaspectratio]{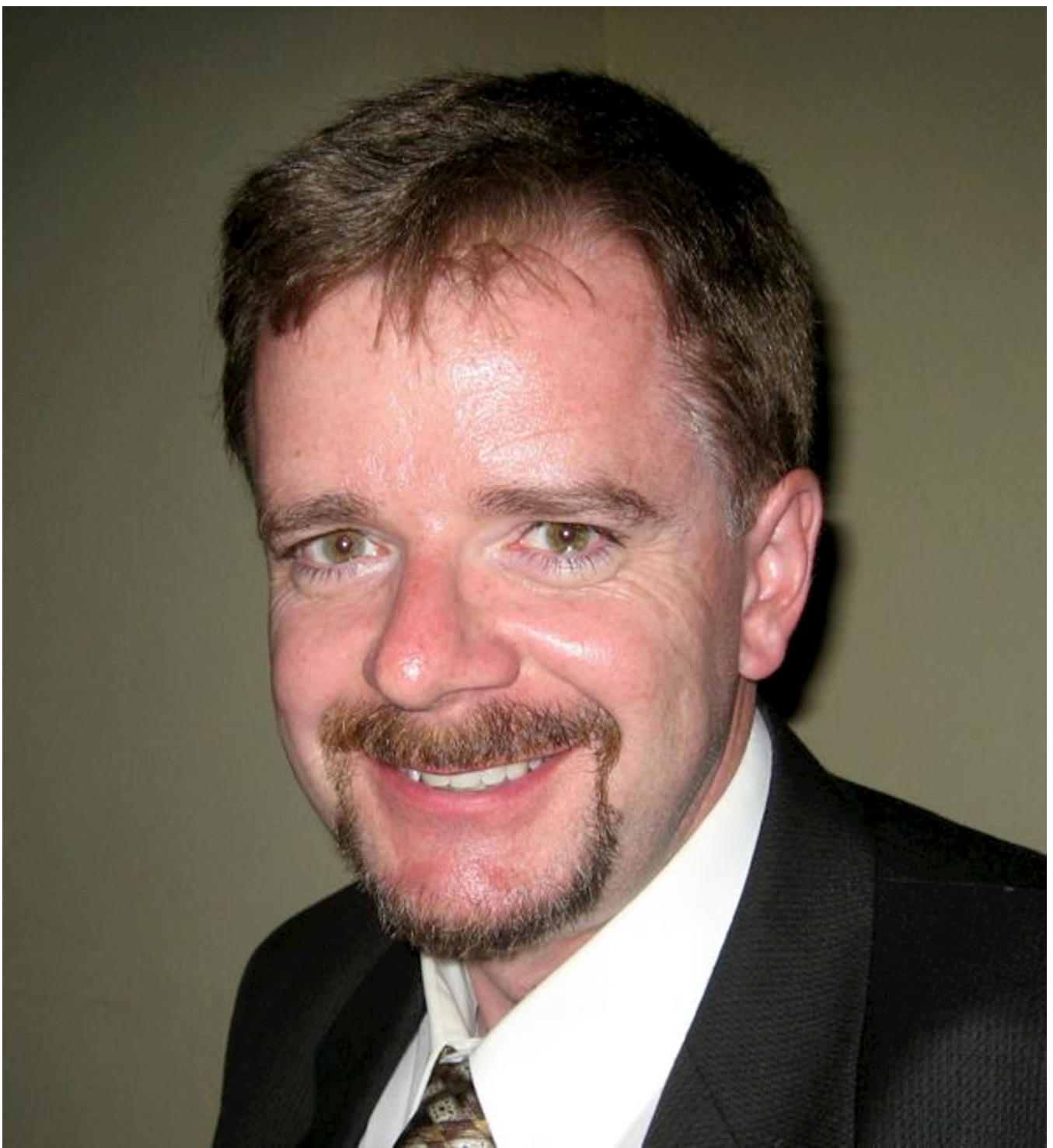}}]{Martin Haenggi}
Martin Haenggi (S-95, M-99, SM-04) is a Professor of Electrical Engineering and a Concurrent Professor of Applied and Computational Mathematics and Statistics at the University of Notre Dame, Indiana, USA. He received the Dipl.-Ing. (M.Sc.) and Dr.sc.techn. (Ph.D.) degrees in electrical engineering from the Swiss Federal Institute of Technology in Zurich (ETH) in 1995 and 1999, respectively. After a postdoctoral year at the University of California in Berkeley, he joined the University of Notre Dame in 2001. In 2007-2008, he spent a Sabbatical Year at the University of California at San Diego (UCSD). For both his M.Sc. and Ph.D. theses, he was awarded the ETH medal, and he received a CAREER award from the U.S. National Science Foundation in 2005 and the 2010 IEEE Communications Society Best Tutorial Paper award.
He served an Associate Editor of the Elsevier Journal of Ad Hoc Networks from 2005-2008, of the IEEE Transactions on Mobile Computing (TMC) from 2008-2011, and of the ACM Transactions on Sensor Networks from 2009-2011, as a Guest Editor for the IEEE Journal on Selected Areas in Communications in 2008-2009 and the IEEE Transactions on Vehicular Technology in 2012-2013, and as a Steering Committee Member for the TMC. Presently he is the chair of the Executive Editorial Committee of the IEEE Transactions on Wireless Communications. He also served as a Distinguished Lecturer for the IEEE Circuits and Systems Society in 2005-2006, as a TPC Co-chair of the Communication Theory Symposium of the 2012 IEEE International Conference on Communications (ICC'12), and as a General Co-chair of the 2009 International Workshop on Spatial Stochastic Models for Wireless Networks (SpaSWiN'09) and the 2012 DIMACS Workshop on Connectivity and Resilience of Large-Scale Networks, and as the Keynote Speaker of SpaSWiN'13. He is a co-author of the monograph "Interference in Large Wireless Networks" (NOW Publishers, 2009) and the author of the textbook "Stochastic Geometry for Wireless Networks" (Cambridge University Press, 2012). His scientific interests include networking and wireless communications, with an emphasis on ad hoc, cognitive, cellular, sensor, and mesh networks.
\end{IEEEbiography}
\end{document}